%% file: main.tex
\documentclass[11pt]{article}

\usepackage[margin=1in]{geometry}
\input{envs}

\usepackage{algorithm}
\usepackage[noend]{algorithmic}

\title{A Hierarchical Grouping Algorithm for the Multi-Vehicle Dial-a-Ride Problem}
\author{
Kelin Luo\\ University of Bonn \\
\texttt{kluo@uni-bonn.de}
        \and 
        Alexandre M. Florio\\ Polytechnique Montreal\\
        \texttt{aflorio@gmail.com} 
        \and Syamantak Das\\  Indraprastha Institute of Information Technology Delhi\\
\texttt{syamantak@iiitd.ac.in} 
 \and Xiangyu Guo\\  University at Buffalo\\
\texttt{xiangyug@buffalo.edu} 
}

\begin{document}
    \maketitle

\maketitle
 \begin{abstract}
  Ride-sharing is an essential aspect of modern urban mobility. In this paper, we consider a classical problem in ride-sharing -- the Multi-Vehicle Dial-a-Ride Problem (Multi-Vehicle DaRP). Given a fleet of vehicles with a fixed capacity stationed at various locations and a set of ride requests specified by origins and destinations, the goal is to serve all requests such that no vehicle is assigned more passengers than its capacity at any point along its trip.    
  We propose an algorithm \HRA, which is the \emph{first non-trivial approximation algorithm} for the Multi-Vehicle DaRP. The main technical contribution is to reduce the Multi-Vehicle DaRP to a certain capacitated partitioning problem, which we solve using a novel hierarchical grouping algorithm. 
  Experimental results show that the vehicle routes produced by our algorithm not only exhibit less total travel distance compared to state-of-the-art baselines, but also enjoy a small in-transit latency, which crucially relates to riders' traveling times. This suggests that \HRA enhances rider experience while being energy-efficient.
\end{abstract}

    \thispagestyle{empty}

\input{intro.tex}

\input{prelim}

\input{algos}
\input{grouping_alg}

\input{routing}
\input{experiments}

\input{scalability}
\input{conclusion}

\bibliography{main}

    \bibliographystyle{alpha}
   
\end{document}

%% file: envs.tex
\usepackage[utf8]{inputenc}
\usepackage{soul}
\usepackage{url}
\usepackage{amsmath,amsfonts,amsthm}
\usepackage{thmtools}
\usepackage{thm-restate} 

\allowdisplaybreaks

\DeclareMathOperator*{\argmin}{arg\,min}
\usepackage{booktabs}
\usepackage{color}
\usepackage{fancybox}
\usepackage{float}
\usepackage{dsfont}
\usepackage{bm}
\usepackage{bbm}
\usepackage{hyperref,cleveref}
\usepackage{subfig}

\usepackage{xspace} 
\usepackage{enumerate}
\usepackage{graphicx}

\usepackage{comment}

\newtheorem{claim}{Claim}
\newtheorem{fact}{Fact}
 \newtheorem{example}{Example} 
 \usepackage{array}
\newtheorem{definition}{Definition}

\newcommand{\brackets}[1]{{\left\langle #1 \right\rangle}}

\newcommand{\bbR}{{\mathbb{R}}}

\newcommand{\mst}{\mathsf{mst}}

\newcommand{\cost}{\mathsf{cost}}
\newcommand{\walk}{\mathsf{walk}}

\newcommand{\bfw}{\mathbf{w}}
\newcommand{\calA}{\mathcal{A}}

\newcommand{\calF}{\mathcal{F}}

\newcommand{\calP}{\mathcal{P}}
\newcommand{\calT}{\mathcal{T}}

\newcommand{\bigO}{{\mathcal O}}
\newcommand{\calK}{{ K }}

\newcommand{\OPT}{\mathsf{OPT}}

\newcommand{\FESI}{\textsf{FESI}\xspace}
\newcommand{\PGP}{\textsf{pruneGDP}\xspace}
\newcommand{\HRA}{\textsf{HGR}\xspace}

\newcommand{\HRAwone}{\textsf{HGR}-$w_1$\xspace}

\newcommand{\HRAapprox}{\textsf{HGR-approx}\xspace}

%% file: intro.tex
\section{Introduction}\label{sec:intro}

Over the last decade, ride-sharing has emerged as one of the most prominent aspects of shared economy~\cite{clewlow2017disruptive}. In a typical ride-sharing scenario, riders with similar routes use a common vehicle for their commutes. The popularity of this framework has soared in recent years owing to the fact that all major urban taxi providers like Uber, Lyft and Didi Chuxing have introduced a `carpooling' option. Economic benefits of ride-sharing are enjoyed by both the riders and providers: riders pay less for the same commute compared to hiring an individual taxi whereas the provider earns more profit in a single ride. Perhaps even more importantly, there is a potentially huge positive impact of ride-sharing on the environment~\cite{cai2019environmental}. Ride-sharing results in overall less fuel consumption and reduces air pollution by decreasing the number of vehicles on the road.

In order to reap the most benefit out of ride-sharing, it is essential to determine an efficient policy of assigning riders to vehicles. Owing to the large scale of the problem and the various constraints it might pose, there has been an increasing body of work in Computer Science and Operations Research that targets to design efficient algorithms to carry out such a task. For instance, there has been significant work focusing on maximizing revenue of the shared mobility provider~\cite{TongZZCYX18-GDP,ZhengCC19,ZhengCY18,ZengTSC21}, minimizing commute distance/time~\cite{BeiZ18-carsharing,luo2020approximation,ZengTC19-LMD} and even optimizing complex social utilities of both servers and requests~\cite{PengHL17}.

In this paper, we consider a classical problem in this area called the Dial-a-Ride Problem (DaRP)~\cite{PaepeLKSSL04}. Informally, the mobility provider has a fleet of vehicles at their disposal, each with a certain capacity. There is a set of ride requests specified by origins and destinations. The algorithmic task is to assign every rider to exactly one vehicle and determine routes for the vehicles under the constraint that at any point during the trip, the vehicle must not accommodate more riders than its capacity. Finally, the goal is to minimize the total travel distance of all the vehicles. 

\vspace{0.5\baselineskip}
\noindent
\textbf{Heuristics for DaRP}. Several heuristic approaches have been proposed for the DaRP over the years (see, for example, the survey~\cite{DaRPSurvey18}). We highlight two recent algorithms which are state-of-the-art and have been experimentally established to be more effective than all the popular heuristics designed previously.  The first one called \PGP was introduced by ~\cite{TongZZCYX18-GDP}.  
This is a fast algorithm that exploits a popular approach called insertion which has been utilized in solving dial-a-ride and its variants~\cite{jaw1984solving,jaw1986heuristic,ma2013tshare,huang2014large,cheng2017utility}. Roughly speaking, the algorithm maintains a partial assignment of requests (and hence routes) for each vehicle. At every iteration, the algorithm determines the assignment of one unassigned request to a vehicle in a way that causes the minimal increment in total travel distance. The authors give an elegant $\bigO(n)$-time implementation of this subroutine and experimentally demonstrate the effectiveness of this heuristic over several previous heuristics like~\cite{huang2014large, ma2013tshare}.

The second algorithm, \FESI~\cite{ZengTC19-LMD}, is an approximation algorithm for the somewhat complementary objective of minimizing the makespan, that is, the maximum travel distance of any vehicle. In fact, the authors claim through empirical evidence that \FESI is comparable to \PGP even for the total travel distance objective although it does not explicitly aims to minimize this. 

Although these algorithms have been experimentally demonstrated to be effective and scalable, none of these works provide a formal worst case performance guarantee on the objective function value of total travel distance. In fact, for both the algorithms, one can easily construct instances where their performance could be arbitrarily bad compared to an optimal solution. 

\vspace{0.5\baselineskip}
\noindent
\textbf{Approximation Algorithms.} There has been significant interest in the theoretical computer science community regarding DaRP. The problem is easily seen to be NP-hard even in the special case when every request has its origin and destination co-located -- this is the classical Travelling Salesman Problem. For the special case of a \emph{single vehicle} with capacity $\lambda$ and $n$ riders, two independent algorithms were given by Charikar and Raghavachari~\cite{charikar1998finite} and later on by Gupta et al.~\cite{gupta2010dial} with approximation guarantees of $\bigO(\sqrt{\lambda}\log n)$ and $\bigO(\sqrt{\lambda}\log^2 n)$, respectively. These are the best known theoretical guarantees so far. However, there is no approximation algorithm reported in the literature for the case of \emph{multiple vehicle} DaRP that we consider.

The above discussion motivates the following question: \emph{Is there an algorithm for multiple vehicle DaRP which is provably good compared to the optimal solution 
in the worst case ?}
In this paper, we give the first non-trivial approximation algorithm for the multiple vehicle DaRP with an approximation ratio of $\bigO(\sqrt{\lambda}\log n)$. Our approximation guarantee, perhaps surprisingly, does not depend on the number of vehicles and exactly matches the guarantee for the single vehicle case stated above. Our technique at a high level resembles the approach used in~\cite{gupta2010dial}. However, we need several non-trivial modifications and novel ideas to handle the multi-vehicle scenario. At the core, our algorithm uses a novel \emph{hierarchical partitioning} of the rider set into groups which can be routed at a `small cost'. These groups are then carefully assigned to vehicles followed by a routing phase for each vehicle. Whereas, for the single vehicle case, such a partitioning can be found by a relatively simple greedy approach, our algorithm needs to heavily utilize bipartite matching and ideas from routing literature which help us to bound from above the total travel distance of our algorithm. Our \textbf{key contributions} are as follows:
\begin{itemize}
    \item We give the \emph{first non-trivial approximation algorithm} for the mutiple vehicle DaRP. Our approximation factor is $\bigO(\sqrt{\lambda}\log n)$, where $\lambda$ is the capacity of the vehicles and $n$ is the number of riders.
    \item Extensive experiments have been carried out to establish the practical efficacy of our algorithm. We compare our algorithm with state-of-the-art heuristics for DaRP. Our method outperforms all these algorithms on total travel distance by a significant margin of up to 30\% on synthetic and real-world datasets.
    \item Our theoretical guarantees are valid only for the objective of minimizing the total travel time of the vehicles. However, in our experiments, we also consider the \emph{in-transit latency of the riders}. This measures the amount of time a rider spends in the vehicle and can be thought of as a metric of rider experience. Empirical evidence shows that our proposed algorithm leads to an average in-transit latency up to 50\% less than other DaRP algorithms.
\end{itemize}

%% file: prelim.tex
\section{Preliminaries}\label{sec:prelim}

\paragraph{Problem definition.} Let $(V, d)$ be a given metric space, $R$ be a set of $n$ requests, where each request $r_i=(s_i,t_i)\in V \times V$ consists of a pickup location $s_i$ and a drop-off location $t_i$. We also have a set of $m$ vehicles $K$, where each vehicle $k\in K$ has a depot $p_k$ and a capacity $\lambda$. Let $V_K$ denote the multiset of all vehicle depot locations.   
The goal is to find an \emph{assignment} $\calA$ from vehicles to requests. An assignment is a collection of \emph{walk}s\footnote{A walk is a finite-length sequence of vertices $v_1,v_2,...,v_N\in V$ for some $N$, and the cost (length) of the walk is defined as $\sum_{i=1}^{N-1}d(v_i,v_{i+1})$.} in $V$, each of which starts from a distinct vehicle depot, and delivers a subset of requests from their pickup locations to drop-off locations.  

\begin{definition}[Vehicle Walk] Given a set of requests $R_k$ assigned to a vehicle $k\in \calK$, a \textbf{vehicle walk} is a sequence $$\walk_k = \brackets{\ell_0= p_k, \ell_1, \ell_2, \cdots, \ell_t},$$ starting at the origin location of vehicle $k$, where $\ell_i\in \{s_r : r\in R_k\} \cup \{t_r : r\in R_k\}, 1\leq i\leq t$. A vehicle walk $\walk_k$ is \textbf{feasible} if \textbf{(i)} $\forall r\in R_k$, $s_r$ appears before $t_r$ in $\walk_k$ and \textbf{(ii)} at any time point of the vehicle walk the corresponding vehicle carries at most $\lambda$ requests. Further, the \textbf{cost} of a walk $\walk_k$ is defined as $\cost(\walk_k) = \sum_{i=0}^{t-1} d(\ell_{i}, \ell_{i+1})$.
\end{definition}
 
A \emph{feasible assignment} should deliver all requests, while ensuring that all vehicle walks are feasible.
The objective is to minimize the total travel distance of all vehicle walks in a feasible assignment $\calA$, denoted as $\cost(\calA)$, i.e. $\min \sum_{\walk_k\in \calA} \cost(\walk_k)$.

\begin{definition}[Multi-Vehicle DaRP]
\label{def:dialarideproblem}
Given a metric space $(V, d)$, 
a set of $n$ requests $R:=\{s_i, t_i\}_{i=1}^{n}\in V^2$, and a set of $m$ vehicles with locations $V_K:=\{p_k\}_{k=1}^{m} \in V$ and a capacity $\lambda$, find a set of \textbf{minimum length vehicle walks} of the vehicles starting at $\{p_k\}_{k=1}^{m} \in V$ that moves each  request $r_i$ from its origin $s_i$
to its destination $t_i$ such that each vehicle
carries at most $\lambda$ requests at any point along the walk.
\end{definition}

We say that a request is preempted if, after being picked up from its origin, it is left temporarily at some vertex before being picked-up again and delivered to its destination. In our setting this is not allowed, as we study the \textbf{non-preemptive} DaRP. Finally, when referring to set of requests, we view it both as a set of pairs in $V\times V$ and as a subset of $V$, where in the latter case it contains all pickup and drop-off locations appearing in the requests. Which viewpoint is being used should be clear from the context.

\begin{example} 
\label{exa:instance}
We use Figure~\ref{fig:example_instance} as a running example. Given $2$ vehicles originally located at $p_1$ and $p_2$, and $8$ customer requests $r_i$ ($i\in [8]$) where customer $i$ aims to travel from $s_i$ to $t_i$; The vehicle capacity is $4$ and the distance metric is denoted as $d$. We would like to find two vehicle walks starting at $\{p_k\}_{k=1}^{2} $ that moves each $r_i$ from its origin $s_i$ to its destination $t_i$ such that each vehicle carries at most $4$ requests at any point along the walk. 
There are  $2^8$ possible ways to assign the $8$ requests to the two vehicles. And a vehicle can have many different orders to serve the assigned requests. For example, if $r_1, r_2$ are assigned to vehicle $1$, then there are $6$ feasible walks: $$\brackets{p_1, s_1, t_1,  s_2, t_2},\brackets{p_1, s_1,   s_2,t_1,t_2}, \brackets{p_1, s_1,  s_2, t_2, t_1},$$
 $$\brackets{p_1, s_2, t_2,  s_1, t_1},\brackets{p_1, s_2,   s_1,t_2,t_1}, \brackets{p_1, s_2,  s_1, t_1, t_2}.$$ 
The cost of serving requests is equal to the total length of the vehicle walks. For example, the cost of a walk $\brackets{p_1, s_1, t_1,  s_2, t_2}$ is equal to  $d(p_1, s_1)+d(s_1, t_1)+d(t_1, s_2) +d(s_2, t_2)$. Then, the optimal solution is the minimum length vehicle walks that serve all requests.
\end{example}

\begin{figure}[htbp]
	\vspace*{-1ex}
	\centerline{\includegraphics[width=0.5\linewidth]{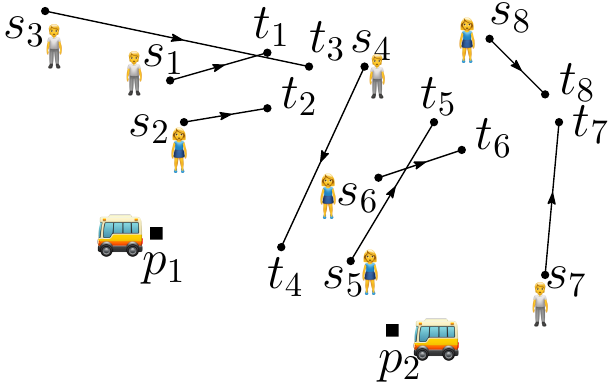}} 
	\caption{DaRP example}
	\label{fig:example_instance}
		 
\end{figure}

%% file: algos.tex
\section{\HRA : A New Algorithm for DaRP}\label{sec:algos} 

In this section, we introduce a novel $O(\sqrt{\lambda}\log n)$-approximation algorithm, which we call \emph{\underline{H}ierarchical \ \underline{G}rouping and \underline{R}outing} (HGR), for the Multi-Vehicle DaRP. 

Before going in to the details of the algorithm and proof of approximation ratio, we give some high level ideas about our main techniques.  
The first idea towards designing this algorithm is to partition the requests into disjoint \emph{groups} of size at most $\lambda$. The intent is - A vehicle starts empty, serves one group entirely before moving on to the next one. Note that an optimal solution to this problem does not necessarily follow this strategy. However, the authors in 
~\cite{gupta2010dial} shows that there always exists a \emph{near-optimal} solution which follows such a strategy. The following Fact~\ref{fact:construction} makes this formal.
\begin{fact}[Solution Structure]{\rm \cite{gupta2010dial}}
\label{fact:construction}
Given any DaRP instance, there exists a feasible walk $\tau$ satisfying the following conditions:
\begin{itemize} 
\item $\tau$ can be split into a set of segments $\{S_1,..., S_t\}$ 
where each segment $S_i$
services a set $O_i$ of at most $\lambda$ requests such that $S_i$
is a path that first picks up each request in $O_i \subseteq R$
and then drops each of them.
\item  The length of $\tau$ is at most $O(\log n)$ times the length of an optimal walk.
\end{itemize}
\end{fact} 

 Although the above fact has been only proven for the single-vehicle DaRP in~\cite{gupta2010dial}, it is not difficult to generalize this to the case of   
 Multi-Vehicle DaRP. The authors in ~\cite{gupta2010dial} exploits the above fact effectively to design a \emph{greedy} algorithm for the single vehicle case which works roughly as follows. The algorithm is iterative, where, at each iteration, a group of $\lambda$ requests (except possibly at the last iteration which might have less than $\lambda$ requests) are formed. The criteria to form a new group is that among the remaining requests, they can be served by travelling the \emph{minimum total distance}. A significant challenge in~\cite{gupta2010dial} was to design an algorithm which makes the above greedy choice in each iteration. This requires them to solve a highly non-trivial problem which they call $\lambda$-forest and which has connections with a notoriously hard problem called the $\lambda$-densest sub-graph problem. Indeed, the main contribution of the above paper was to give a $O(\sqrt{\lambda})$-approximation algorithm for $\lambda$-forest. This was combined with a standard argument from approximation algorithms literature to show that the overall approximation guarantee is $O(\lambda\log^2 n)$. However, this approach poses the following challenges in being effective as a practical algorithm for the multiple-vehicle case.
 
 \begin{enumerate}
     \item The algorithm they use to solve $\lambda$-forest,  although giving reasonable approximation guarantees, is highly complicated and not practical. In fact, we had implemented this approach for a single vehicle and found the running time scaling prohibitively with the number of requests (For 200 requests and capacity 16, it takes 1900 seconds, and for 500 requests and capacity 4, it takes 1500 seconds. We contrast this with our algorithm can handle say 10k requests with capacity 32 in about 100 seconds).
     
     \item It is not immediately clear how to adopt the above approach to the multiple vehicle case. Although Fact~\ref{fact:construction} still continues to be true, a major challenge here is to determine which group to assign to which vehicle (note that this issue does not exist in the single vehicle case). So, any algorithm which aims to provide a theoretical guarantee needs to combine the grouping phase and the assignment phase while not incurring a lot of cost.
 \end{enumerate}
 In order to overcome these two issues, we avoid the \emph{local greedy} approach of~\cite{gupta2010dial} and develop a novel algorithm that exploits a more \emph{global} viewpoint. 
 Our \emph{core technical contribution} is to design a \emph{hierarchical grouping technique} which avoids solving complicated problems like $\lambda$-forest. Instead we use relatively simpler sub-routines like bipartite matching and minimum spanning trees and still manage to obtain the same approximation ratio as~\cite{gupta2010dial} for the multiple vehicle case. This not only makes our algorithm more efficient and relatively easier to implement, we are also able to avoid the additional $\log n$ factor incurred by~\cite{gupta2010dial} due to the iterative greedy approach.
 Now we present the technical details of our algorithm. 
 As mentioned above, our main idea is to develop a hierarchical clustering algorithm to solve the grouping problem. Indeed, we partition requests into groups of size $\le \lambda$ by considering the following capacitated grouping problem (see Definition~\ref{def:groupingproblem}), and give an approximation guarantee of $O(\sqrt{\lambda})$ (see Theorem~\ref{thm:partition} in Section~\ref{sec:alg_grouping}).

\begin{definition}[Capacitated Grouping Problem]
\label{def:groupingproblem}
Given an $n$-vertex metric space $(V, d)$ and requests $R:=\{s_i, t_i\}_{i=1}^{m}\in V^2$, find a set of \textbf{minimum length walks} that serves all requests and such that each walk covers at most $\lambda$ requests.
\end{definition} 

When talking about a walk in the capacitated grouping problem, it is always associated with the request group covered by it. 
A \emph{feasible partition} $\calP$ of $R$ partitions $R$ into groups of size at most $\lambda$.
A \textbf{walk} covering request group $P \in \calP $ is a sequence $\bfw_P = \brackets{\ell_1, \ell_2, \cdots, \ell_h}$ that traverses $P$, where $\ell_i\in \{s_r : r\in P\} \cup \{t_r : r\in P\}, 1\leq i\leq h$. A walk $\bfw_P$ is \emph{feasible} if $\forall r\in P$, $s_r$ appears before $t_r$ in $\bfw_P$. The cost of  $\bfw_P$ is denoted as $\cost(\bfw_P)=\sum_{i=1}^{h-1} d(\ell_i, \ell_{i+1})$.  The cost of partition $\calP$ is denoted as $\cost(\mathcal{P})= \sum_{P\in\calP}\min_{\text{feasible }\bfw_P}\cost(\bfw_P)$.

Our main algorithm (Algorithm~\ref{alg:main-algorithm}) for DaRP first treats the input as an instance of Capacitated Grouping Problem and solves it using Algorithm~\ref{alg:partition_notexact} to get a partition, then builds an actual route based on the partition. We now describe the main ideas in each step  (See an example in  Figure~\ref{fig:example_algorithm}).
  
 \vspace{-1mm}
\begin{algorithm}
\begin{algorithmic}[1]
\small
\REQUIRE Request set $R$, Vehicle locations $V_K$ and capacity $\lambda$
\ENSURE 
A feasible assignment $\mathcal{A}$ 

\STATE $\mathcal{P} \gets$ \textsc{Hierarchical Grouping}$(R, \lambda)$ \COMMENT{Alg.~\ref{alg:partition_notexact}} 
    
\STATE $\mathcal{A}\gets$ \textsc{Routing}($R, \calP, V_K$) \COMMENT{Alg.~\ref{alg:assignandroute}} 
   
   \RETURN   $\mathcal{A}=\{\walk_k: k\in \calK\}$
 \caption{\small{\textsc{Hierarchical Grouping and Routing} (HGR)}}
 	\label{alg:main-algorithm}
\end{algorithmic}
\end{algorithm}

 \vspace{-1mm}

In Step 1, the \textsc{Hierarchical Grouping} algorithm partitions requests into groups such that the total length of walks covering the partition is not too large compared with the optimal solution. To achieve this goal, we develop a non-trivial two-layer hierarchical grouping technique: in the outer layer, iteratively combine two clusters of requests into one cluster; inside each cluster, we form groups to ensure that closer requests are grouped together and far-apart requests are divided into separate groups. 
 


In Step 2, we use the partition obtained in Step 1 to design actual routes for all vehicles. The idea is to view each group as a single vertex, and compute a cheap spanning forest to assign vehicles to groups. The forest is computed such that each tree in it is rooted as some vehicle location of $V_K$, and this vehicle will traverse the tree to serve its requests in a group-by-group manner. We obtain the following result. 

\begin{restatable}{theorem}{ThmApproxDaRP}
\label{thm:approx-darp}
 Given a Multi-Vehicle DaRP with set of requests $R$ and set of vehicles $\calK$, each with a capacity $\lambda$, the HGR algorithm runs in time $\bigO(|R|^3\log \lambda+ |R|^2 \lambda^2 \log \lambda)$ and returns a set of $|\calK|$ feasible walks serving all requests in $R$ such that the total travel distance is at most $\bigO(\sqrt\lambda \cdot 
 \log |R|)$ times that of an optimal solution to the Muil-Vehicle DaRP. 
\end{restatable}

\begin{example} 
\label{exa:instance_algorithm}
Figure~\ref{fig:example_algorithm} shows running HGR on the Multi-DaRP instance in Example~\ref{exa:instance} ( Figure~\ref{fig:example_instance}). HGR first invokes \textsc{Hierarchical Grouping} (Alg.~\ref{alg:partition_notexact}) to partition requests into groups. In the outer layer of HG, we form clusters hierarchically. Here, there are $4$ clusters in the 1-st iteration: $\{ \{r_1, r_2\} \}$,  $\{ \{ r_3\}, \{ r_4\}\}$, $\{ \{r_5, r_6\} \}$ and $\{ \{ r_7, r_8\}\}$, and $2$ clusters in the $2$-nd iteration: $\{ \{r_1, r_2, r_3\}, \{ r_4\}\}$, $\{ \{r_5, r_6, r_7, r_8\}\}$.
In the inner layer, we form groups inside each cluster based on the requests. 
If the requests of two clusters in the previous iteration are far-apart, although they are combined to one cluster, they are divided into separate groups. See example of the two clusters $\{\{r_3\}\}, \{\{r_4\}\}$ in $0$-th iteration and the combined cluster $\{\{r_3\}, \{r_4\}\}$ in $1$-th iteration in Figure~\ref{fig:example_algorithm}. If some requests of two clusters in the previous iteration are close and they are combined to one cluster, then they are grouped together. See example of  the two clusters $\{\{r_5, r_6\}\}, \{\{r_7, r_8\}\}$ in $1$-th iteration and the combined cluster $\{\{r_5, r_6, r_7, r_8\}\}$ in $2$-nd iteration in Figure~\ref{fig:example_algorithm}. After $\log 4=2$ iterations, we obtain the partition of the requests: $\{r_1, r_2, r_3\}, \{r_4\}, \{r_5, r_6, r_7, r_8\} $.

In the second step, HGR invokes \textsc{Routing} (Alg.~\ref{alg:assignandroute}) to build the final route based on the obtained partition. Each request group (contains no more than $\lambda$ requests) is viewed as a point, and we compute a minimum spanning forest over $7$ points $ p_1, p_2, \{r_1, r_2, r_3\}, \{r_4\}, \{r_5, r_6, r_7, r_8\} $ such that each tree is rooted at $p_1, p_2$. Then by traversing the tree to serve requests in a group-by-group manner, we obtain the two vehicle walks $\brackets{p_1, s_3, s_1, s_2, t_1, t_2, t_3, s_4, t_4},\brackets{p_2, s_5, s_6, s_7, t_6, t_5, s_8, t_8, t_7}.$
\end{example}

\begin{figure}[!htbp]
	\centerline{\includegraphics[width=0.8\linewidth]{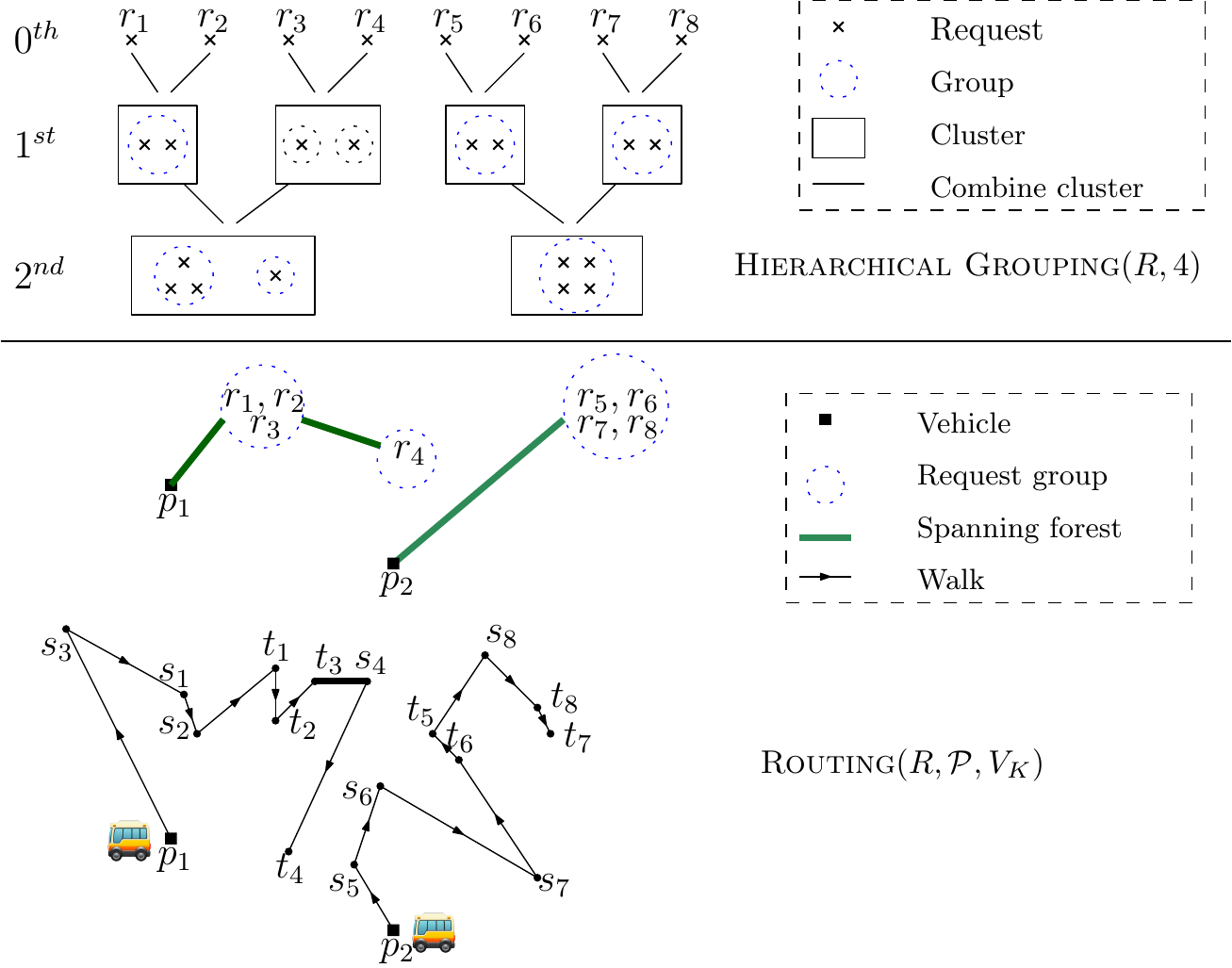}} 
	\caption{HGR for the DaRP example}
	\label{fig:example_algorithm}
			
\end{figure}
 
We describe all steps of the HGR algorithm in more detail 
in Section~\ref{sec:alg_grouping} and Section~\ref{sec:assignandroute}. 

%% file: grouping_alg.tex
\section{Part (I): Grouping}
\label{sec:alg_grouping}

In this section, we present the \textsc{Hierarchical Grouping} (HG) algorithm, and use it to give an approximation guarantee
of $O(\sqrt{\lambda})$ for the capacitated grouping problem.   The intuition behind HG is to cluster requests iteratively such that in each iteration, the number of requests clustered together is doubled, and after $\log \lambda$ iterations, every cluster contains $\lambda$ requests. However, treating each cluster as a single group does not necessarily give small cost (i.e., travel length): In fact, this can be far from optimal, which may use much smaller groups.  Therefore,  
we develop a non-trivial two-layer hierarchical grouping technique as follows:  In the outer layer, we form sub-clusters hierarchically such that by combining $2$ clusters of the $(i-1)$-th iteration, the number of requests in each cluster is $2^i$ in the $i$-th iteration, until each cluster contains $\lambda$ requests\footnote{It also works for $\lambda$ which is (possibly) not exact powers of 2: for an arbitrary $\lambda$ with $2^i<\lambda< 2^{i+1}$, the process ends when each cluster contains $2^i$ requests. It does not affect the analysis.};  In the inner layer, we form groups inside each cluster to ensure that closer requests are grouped together and far-apart requests are divided into separate groups.  

The key idea is to define a suitable edge cost on a graph whose vertices represent the clusters formed in each iteration, and each iteration is carried out by computing a minimum-\emph{cost} matching on this graph. Notice that each iteration doubles the number of requests in each cluster and initially each cluster contains exactly one request. Since at the end each cluster contains no more than $\lambda$ requests, the hierarchical grouping process will continue for $\lfloor \log \lambda \rfloor$ iterations. The design of the edge cost function of the graph needs to be very careful, such that we can bound the cost of covering the groups with respect to the optimal cost.

Before formally presenting the algorithm,  we define the following notations:  
 
\begin{itemize}
\item (Request-) group $P, X\subseteq R$: a set of requests

\item (Group-) cluster $\mathcal{P,Q,X}$: a set of  groups

\item (Cluster-) collection $\mathfrak{M}$: a set of clusters

\end{itemize}
 
Specifically, we use notation $w(\mathcal{X}, \mathcal{X}')$  to represent \textbf{cost function} on the edges (connecting two clusters $\mathcal{X}$ and $\mathcal{X}'$) of a graph. 
We first define the following notations for (request-) groups $X, X'\subseteq R$:    
\begin{enumerate}[(1)]

\item Minimum Spannig Tree (MST) cost  $\mst_{s,t}(X)$:  
$\mst_s(X)$ (resp. $\mst_t(X)$) is defined to be the cost of a MST over the origins (resp. destinations) of all requests in $X$. Further, define $\mst_{s,t}(X)=\mst_s(X)+\mst_{t}(X)$.

\item  (Incremental) Cost of serving groups together: $w_1(X, X')=\mst_{s,t}(X\cup X')-\mst_{s,t}(X)-\mst_{s,t}(X')$.

\item  Cost of serving groups separately: $w_2(X,X')=\min_{r_i\in X}  d(s_i,t_i)+\min_{r_i\in X'} d(s_i,t_i) $.
\end{enumerate}

During the execution of HG, we build a larger cluster by merging smaller clusters, and we can choose to either merge the groups \emph{within} the smaller clusters or  leave them separated. To guide the choice, we define the \textbf{cost function} of merging two clusters as follows: for every group pair $X\in \mathcal{X}, X'\in \mathcal{X}'$, we compare the (incremental) cost $w_1(X, X')$ of serving two groups together and the cost $w_2(X, X')$ of serving groups separately; then, we take the minimum as the cost of combining the two clusters (See Definition~\ref{def:costfunction}).

\begin{definition}[Cost Function]
\label{def:costfunction}
Given two clusters $\mathcal{X}$ and $\mathcal{X}'$, define 
$$w(\mathcal{X},\mathcal{X}')=\min_{X\in \mathcal{X}, X'\in \mathcal{X}'}  \min \{w_1(X,X'),w_2(X,X')\}$$ 
where $w_1$ and $w_2$ are defined as above. 
\end{definition}

We now describe the algorithm formally. 
In order to simplify the description, we assume for now that $\lambda$ is a power of 2. It is straightforward to adapt the procedures for arbitrary $\lambda$ with the same approximation guarantee.

In Algorithm~\ref{alg:partition_notexact}, we iteratively build larger clusters in a hierarchical way, until every cluster is of size $\lambda$. Note that each cluster may still contain multiple request groups.   
We first initialize a trivial collection $\mathfrak{M}_0$ containing $|R|$ clusters, where each cluster contains one group and each group contains a distinct request (Line 1).  Then, we repeat for $\log\lambda$ iterations, where in each iteration we compute a minimum weight perfect matching \footnote{ The perfect matching is a well-studied problem. In our computational experiments, we use a modern implementation of Edmonds’s algorithm, available at~\cite{kolmogorov2009blossom}.}  on the current cluster collection (Line 7-8), and merge the matched clusters to form a collection for the next iteration (Line 9-13). 
Therefore, after each iteration, the number of clusters is halved, while the number of requests in each cluster is doubled.  
The weight function $w$ defined in Definition~\ref{def:costfunction} is crucial for guiding the merge step (Line 9-13). When merging two clusters, we may also merge two groups from them (Line 11), or leave the groups untouched (Line 13), depending on whether $w$ achieves its value via $w_1$ or $w_2$. Finally, we ``unbox'' the clusters and return all the groups formed.
   
\begin{algorithm}
\begin{algorithmic}[1]
\small
 \caption{\textsc{Hierarchical Grouping}($R$,$\lambda$)}
 	\label{alg:partition_notexact}
\REQUIRE $R$ and $\lambda$
\ENSURE a partition $\mathcal{P}$, each group $P\in\mathcal{P}$ contains not more than $\lambda$ requests
\STATE  $\mathfrak{M}_0=\bigcup_{r_i\in R} \{ \{\{r_i\}\}\}$ 
   
\STATE  $\ell=0$
  
\WHILE{$\ell< \log \lambda$}
        
    \STATE $\ell=\ell+1$
        
    \STATE $\mathfrak{M}_\ell=\emptyset$
        
    \STATE $\mathcal{P}_\ell=\emptyset$
        
    \STATE Let $G_\ell \equiv (\mathfrak{M}_{\ell-1}, e)$ be a complete graph with edge weights $w(\mathcal{Q} ,\mathcal{Q}')$ for any $ \mathcal{Q}, \mathcal{Q}'\in \mathfrak{M}_{\ell-1}$
         
    \STATE Find a minimum weight matching $M_\ell$ in  $G_\ell \equiv (\mathfrak{M}_{\ell-1}, e)$ with total weight $w(M_\ell)=\sum_{(\mathcal{Q}, \mathcal{Q}')\in M_\ell} w(\mathcal{Q}, \mathcal{Q}')$ \label{step:cluster-matching}

    \FOR{$(\mathcal{Q}, \mathcal{Q}')\in M_\ell$}
          
        \IF{$w(\mathcal{Q}, \mathcal{Q}')= w_1(P, P'), P\in \mathcal{Q}, P'\in \mathcal{Q}'$ } 
            \STATE   add cluster $ \mathcal{Q} \cup  \mathcal{Q}' \cup \{P \cup P'\} \setminus \{P, P'\} $ to collection $\mathfrak{M}_i $  
        \ELSIF{ $w(\mathcal{Q}, \mathcal{Q}')= w_2(P,P'), P\in \mathcal{Q}, P'\in \mathcal{Q}'$ }
            \STATE  add cluster $ \mathcal{Q} \cup  \mathcal{Q}' $ to collection $\mathfrak{M}_i $  
        \ENDIF
    \ENDFOR
 \ENDWHILE
   \STATE $\mathcal{P}_l\gets\{Q\in \mathcal{Q}: \mathcal{Q}\in \mathfrak{M}_{l}\}$ for $l=1,\ldots,\log\lambda$
   
    \RETURN $w(M_\ell)$ for $1\le \ell \le \log \lambda$ and $\mathcal{P}= \mathcal{P}_{\log \lambda}$
    \end{algorithmic}
\end{algorithm}

To prove our main result Theorem~\ref{thm:partition},   We first prove some important properties of the \textsc{Hierarchical Grouping} (HG) algorithm that will help us to bound total cost of the capacity-bounded groups. 

Let $\mathcal{P}_l$ denote the $l$-th partition obtained by the HG Algorithm~\ref{alg:partition_notexact} and let $ \mathcal{P}_{l,=i}\subseteq \mathcal{P}_{l}$ denote the request groups of size $2^i$, i.e.,  $\mathcal{P}_{l,=i}=\{P\in \mathcal{P}_l: |P|=2^i\}$, and let $\mathcal{P}_{l,<i}$ (resp. $\mathcal{P}_{l,>i} $) denote the request groups of size smaller than  (resp. more than) $2^i$, i.e,  $\mathcal{P}_{l, < i}=\{P\in \mathcal{P}_l: |P|< 2^i\}$,  $\mathcal{P}_{l,>i}=\{P\in \mathcal{P}_l: |P|>2^i\} $.  
Based on the definition of the cost function $w$, we have the following lemma by a telescope sum:  
 
\begin{restatable}{lemma}{telescopesum}
\label{lemma:partitionalg_each} 
For partition $\mathcal{P}_l$, $ l\in [\log \lambda]$,  obtained by the algorithm~\ref{alg:partition_notexact}, we have
\begin{align*}
\sum_{P\in \mathcal{P}_l} \left(\mst_{s,t}(P)+ \min_{r_i\in P} d(s_i,t_i) \right)\le  
\sum_{i=1}^{l} w(M_i) + \sum_{P\in \mathcal{P}_{l, =l}} \min_{r_i\in P} d(s_i,t_i).
\end{align*}
\end{restatable}
 
 \begin{proof}
We first claim that 
\[\sum_{i=1}^{l} w(M_i) \ge \sum_{P\in \mathcal{P}_l} \mst_{s,t}(P)+  \sum_{P\in\mathcal{P}_{l, <l}} \min_{r_i\in P} d(s_i,t_i) \]
If this claim holds, the lemma follows directly since $|P|\le 2^l$ for all $P\in \mathcal{P}_l$.

We then prove the above claim for every $l\ge 1$. 

When $l=1$, according to the definition of $w(\cdot ,\cdot)$, for any two clusters  $\mathcal{Q}=\{\{r_{q}\}\},\mathcal{Q}'=\{\{r_{q'}\}\}\in \mathfrak{M}_0$,  $$w(\mathcal{Q},\mathcal{Q}')=\min\{d(s_q,s_{q'})+d(t_q,t_{q'}), d(s_q,t_{q})+d(s_{q'},t_{q'})\}.$$  

There are two kinds of  $(\mathcal{Q},\mathcal{Q}')\in M_1$: 
\begin{itemize}
    \item   
$w(\mathcal{Q},\mathcal{Q}')=w_2(\{r_q\},\{r_{q'}\})$. Observe that $ \mst_{s,t}(\{r_q\})=\mst_{s,t}(\{r_{q'}\})=0$. 
For such pair $(\mathcal{Q},\mathcal{Q}')$, the edge weight is  
$$  \mst_{s,t}(\{r_q\})+ \min_{r_i\in \{r_{q}\}} d(s_i,t_i)  + \mst_{s,t}(\{r_{q'}\})+ \min_{r_i\in \{r_{q'}\}} d(s_i,t_i)= w(\mathcal{Q},\mathcal{Q}').$$

\item  
$w(\mathcal{Q},\mathcal{Q}')=w_1(\{r_q\},\{r_{q'}\})$. For such pair $(\mathcal{Q},\mathcal{Q}')$, the edge weight is  $ \mst_{s,t}(\{r_q,r_{q'}\})=d(s_q,s_{q'})+d(t_q,t_{q'}) = w(\mathcal{Q},\mathcal{Q}').$

\end{itemize}
To sum up all items $(\mathcal{Q},\mathcal{Q}')\in M_1$, we have
\begin{align} w(M_1) =& \sum_{(\mathcal{Q},\mathcal{Q}')\in M_1} w(\mathcal{Q},\mathcal{Q}') \nonumber  \\
=& \sum_{P\in \mathcal{P}_1} \mst_{s,t}(P)+  \sum_{P\in\mathcal{P}_{1,<1}} \min_{r_i\in P} d(s_i,t_i)   
\label{ieq:sum_1_partition}
\end{align}

Considering $\mathcal{P}_{l+1}$ with $l\ge 1$, we only need to prove 
 \begin{align}
 w({M}_{l+1}) \ge &\sum_{P\in \mathcal{P}_{l+1}} \mst_{s,t}(P)+  \sum_{P\in\mathcal{P}_{l, <l+1}} \min_{r_i\in P} d(s_i,t_i) \nonumber \\
 &-\left(\sum_{P\in \mathcal{P}_l} \mst_{s,t}(P)+  \sum_{P\in\mathcal{P}_{l,<l}} \min_{r_i\in P} d(s_i,t_i)\right).   \label{ieq:diff_partition}
\end{align}

There are six cases about $(\mathcal{Q},\mathcal{Q}') \in M_{l+1}$ in the $(l+1)$-th iteration: (a) 
For the first three cases such that $w(\mathcal{Q},\mathcal{Q}')=w_1(P,P')$ where $P\in \mathcal{Q}$ and $P'\in \mathcal{Q}'$: Comparing to $\mathcal{P}_l$, the algorithm adds request group $P\cup P'$ to  $\mathcal{P}_{l+1}$ and $P,P' \notin \mathcal{P}_{l+1}$, while $P,P' \in \mathcal{P}_{l}$ and $P\cup P' \notin \mathcal{P}_{l}$; Note that $$w(\mathcal{Q},\mathcal{Q}')=\mst_{s,t}(P\cup P')-\mst_{s,t}(P)-\mst_{s,t}(P'). $$ (b) For the remaining three cases such that $w(\mathcal{Q},\mathcal{Q}')=w_2(P,P')$ where $P\in \mathcal{Q}$ and $P'\in \mathcal{Q}'$:  Comparing to $\mathcal{P}_l$, the algorithm does not make changes for items of $\mathcal{Q},\mathcal{Q}'$, i.e., request group $P, P' \in \mathcal{P}_{l+1}$ and also $P,P'\in \mathcal{P}_{l}$; Note that $$w(\mathcal{Q},\mathcal{Q}')=\min_{r_i\in P} d(s_i,t_i)+\min_{r_i\in P'} d(s_i,t_i). $$   
   
\begin{enumerate}[(1)]

    \item $|\mathcal{Q}|=|\mathcal{Q}'|=1$ and $w(\mathcal{Q},\mathcal{Q}')=w_1(P,P')$ where $P\in \mathcal{Q}$ and $P'\in \mathcal{Q}'$ .
   
   The right side of Inequality~(\ref{ieq:diff_partition}) is equal to  $\mst_{s,t}(P\cup P')-\mst_{s,t}(P)-\mst_{s,t}(P') $, which is equal to $w(\mathcal{Q},\mathcal{Q}')$.

        \item $|\mathcal{Q}|=1$, $|\mathcal{Q}'|\ge 2$ and $w(\mathcal{Q},\mathcal{Q}')=w_1(P,P')$ where $P\in \mathcal{Q}$ and $P'\in \mathcal{Q}'$.
          
   The right side of Inequality~(\ref{ieq:diff_partition}) is equal to $\mst_{s,t}(P\cup P')-\mst_{s,t}(P)-\mst_{s,t}(P')-\min_{r_i\in P'} d(s_i,t_i) +\min_{r_i\in P\cup P'} d(s_i,t_i)\le \mst_{s,t}(P\cup P')-\mst_{s,t}(P)-\mst_{s,t}(P')$, which is equal to $w(\mathcal{Q},\mathcal{Q}')$. 
         
          \item $|\mathcal{Q}|\ge 2$, $|\mathcal{Q}'|\ge 2$ and $w(\mathcal{Q},\mathcal{Q}')=w_1(P,P')$ where $P\in \mathcal{Q}$ and $P'\in \mathcal{Q}'$.
       
      The value about $\mathcal{Q}$ and $\mathcal{Q}'$ in the right side of Inequality~(\ref{ieq:diff_partition}) is equal to $\mst_{s,t}(P\cup P')-\mst_{s,t}(P)-\mst_{s,t}(P')-\min_{r_i\in P} d(s_i,t_i) -\min_{r_i\in P'} d(s_i,t_i) +\min_{r_i\in P\cup P'} d(s_i,t_i)\le  \mst_{s,t}(P\cup P')-\mst_{s,t}(P)-\mst_{s,t}(P')$, which is equal to $w(\mathcal{Q},\mathcal{Q}')$.

   \item  $|\mathcal{Q}|=|\mathcal{Q}'|=1$ and $w(\mathcal{Q},\mathcal{Q}')=w_2(P,P')$ where $P\in \mathcal{Q}$ and $P'\in \mathcal{Q}'$.
   
   The right side of Inequality~(\ref{ieq:diff_partition}) is equal to $ \min_{r_i\in P} d(s_i,t_i)+\min_{r_i\in P'} d(s_i,t_i)$, which is equal to $w(\mathcal{Q},\mathcal{Q}')$. 
    
          \item $|\mathcal{Q}|=1$, $|\mathcal{Q}'|\ge 2$ $w(\mathcal{Q},\mathcal{Q}')=w_2(P,P')$ where $P\in \mathcal{Q}$ and $P'\in \mathcal{Q}'$.

     The right side of Inequality~(\ref{ieq:diff_partition}) is equal to $ \min_{r_i\in P} d(s_i,t_i)\le \min_{r_i\in P} d(s_i,t_i)+\min_{r_i\in P'} d(s_i,t_i)$, which is equal to $w(\mathcal{Q},\mathcal{Q}')$.
      
          \item $|\mathcal{Q}|\ge 2$, $|\mathcal{Q}'|\ge 2$ $w(\mathcal{Q},\mathcal{Q}')=w_2(P,P')$ where $P\in \mathcal{Q}$ and $P'\in \mathcal{Q}'$. 
    
    The right side of Inequality~(\ref{ieq:diff_partition}) is equal to $0\le \min_{r_i\in P} d(s_i,t_i)+\min_{r_i\in P'} d(s_i,t_i)$, which is equal to $w(\mathcal{Q},\mathcal{Q}')$. 
    
\end{enumerate}

To sum up all items $(\mathcal{Q},\mathcal{Q}')\in M_{l+1}$, Inequality~(\ref{ieq:diff_partition}) holds. 

Combing both Inequality~(\ref{ieq:sum_1_partition}) and~(\ref{ieq:diff_partition}), sum up $w(M_i)$  over $i\in [l]$, the claim holds, and thus the lemma is proved.  
\end{proof}

Next, we will bound  the separate serving cost $ \sum_{P\in \mathcal{P}} \min_{r_i\in P} d(s_i, t_i)$ (see Lemma~\ref{lemma:partitionalg_lastweight}) and $\sum_{i} w(M_i)$, respectively. 
 We further introduce the following notations. Fix an optimal partition  $\mathcal{P}^*$.  
 
\begin{restatable}[The Separate Serving Cost]{lemma}{TheSeparateServingCost} 
\label{lemma:partitionalg_lastweight}
The costs of serving a request separately in each group of $\mathcal{P}_{l,=\log \lambda}$
\[\sum_{P\in \mathcal{P}_{l, =\log \lambda}} \min_{r_i\in P} d(s_i,t_i)  \le      \sum_{P\in \mathcal{P}^{\star}} \max_{r_i\in P} d(s_i,t_i). \]
\end{restatable} 
\begin{proof}
Recall that $\mathcal{P}_{l,=\log \lambda} = \bigcup_{P\in \mathcal{P}, |P|=\lambda} P $. Consider a bipartite graph, $(\mathcal{P}_{l,=\log \lambda}, \mathcal{P}^{\star}, E)$ where, for each request $r\in \bigcup_{P\in \mathcal{P}_{l,=\log \lambda}} P$, there is an edge $e(r_i)$ between vertex $P\in \mathcal{P}_{l,=\log \lambda}$ and $  \mathcal{P}^{\star}$. According to Hall's marriage theorem, we can find a  matching $M$ among $\bigcup_{P\in \mathcal{P}_{l,=\log \lambda}} P$ and $\mathcal{P}^{\star}$ that covers all $ \mathcal{P}_{l,=\log \lambda}$. Let $E(M)$ be the edges in $M$. We have  $\sum_{P\in \mathcal{P}_{l,=\log \lambda}} \min_{r_i\in P} d(s_i,t_i) \le  \sum_{e(r_i)\in E(M)} d(s_i,t_i)  \le      \sum_{P\in \mathcal{P}^{\star}} \max_{r_i\in P} d(s_i,t_i).$
\end{proof} 

\begin{restatable}[The $\ell$-th Grouping Cost]{lemma}{PartitionAlgOneWeight}
\label{lemma:partitionalg_oneweight}  
The grouping cost in $\ell$-th iteration is 
\[w(M_l)  \le  \sum_{P\in \mathcal{P}^{\star}} 2^{\frac{\log \lambda-l}{2}} \mst_{s,t}(P)+  \sum_{P\in \mathcal{P}^{\star}}  \max_{r_i\in P} d(s_i,t_i)\] 
if $ \log \lambda-l$  is even; otherwise,  
\[w(M_l)  \le \frac{3}{2}\cdot  \sum_{P\in \mathcal{P}^{\star}} 2^{\frac{\log \lambda-l-1}{2}} \mst_{s,t}(P)+  \sum_{P\in \mathcal{P}^{\star}}  \max_{r_i\in P} d(s_i,t_i).\] 
\end{restatable}

To prove this lemma, first we give a \emph{Bi-partition rule} for a request group, which is used to partition the request groups in the optimal solution $\mathcal{P}^{\star}$.


\emph{Bi-partition rule:} For a request group $P$ and $l\ge 1$ with $|P|> 2^{l-1}$, suppose  $h(P,l)$ is the unique integer satisfying $h(P,l)-1< \log \lceil \frac{|P|}{2^{l}}  \rceil \le h(P,l) $, we bi-partition request group $P$ into $\lceil \frac{|P|}{2^{l}}  \rceil  $ groups where $\lfloor \frac{|P|}{2^{l}}  \rfloor  $ groups each include exactly $2^{l}$ requests and $  \lceil \frac{|P|}{2^{l}}  \rceil-\lfloor \frac{|P|}{2^{l}}  \rfloor $ group include not more than $2^{l}$ requests. 

Specifically, Let  $\mathcal{Q}_0=\{P\}$. The bi-partition follows: In the $j$-th ($1\le j\le h(P,l)$) partition, we bi-partition each item with size $ 2^{l}\cdot 2^{h(P,l)-j+1}$ of $\mathcal{Q}_{j-1}$ into two subgroups and add them into $\mathcal{Q}_{j}$, each contains $ 2^{l}\cdot 2^{h(P,l)-j} $ requests; we bi-partition the item $X \in \mathcal{Q}_{j-1}$ with size $ 2^{l}\cdot 2^{h(P,l)-j} < |X| <  2^{l}\cdot 2^{h(P,l)-j+1} $ into two subgroups and add them into $\mathcal{Q}_j$, one subgroup with size $ 2^{l}\cdot 2^{h(P,l)-j} $ and the other subgroup with size $  |X|- 2^{l}\cdot 2^{h(P,l)-j}  $ ($< 2^{l}\cdot 2^{h(P,l)-j} $). 
If $j$ is odd, we partition request group $X$ into two subgroups $X_l$ and $X_r$ such that $ \mst_s(X_l)+\mst_s(X_r)\le \mst_s(X)$; if $j$ is even, we partition request group $X$ into two subgroups $X_l$ and $X_r$ such that $ \mst_t(X_l)+\mst_t(X_r)\le \mst_t(X)$.

Note that there always exists a partition satisfying the above rules since there exists a feasible partition such that the edge weight of original tree is equal to the sum of the edge weight of two sub-trees each containing a portion of nodes. If the number of nodes in the tree  $\mst_s(X)$ (or $\mst_t(X)$) is even, thus we can always find an edge in which each vertex connects a portion of nodes; By cutting this edge, we obtain a feasible partition. 

Note that according to \emph{Bi-partition rule},  we partition $P$ into subgroups, each including not more than $ 2^{l}$ requests. 
We claim that after every two partitions, the sum of the edge weight of trees is less than doubled, i.e, for $j\ge 0$ \begin{align}\sum_{X\in \mathcal{Q}_{j+2}} \mst_{s,t}(X)\le 2 \sum_{X\in \mathcal{Q}_j} \mst_{s,t}(X). \label{ieq:OPT_bipar}
\end{align}

To prove the claim: 
According to the partition rule, for each odd $i$ we know \[ \sum_{X\in \mathcal{Q}_{j+1}} \mst_{s}(X)\le \sum_{X\in \mathcal{Q}_j} \mst_{s}(X) \] and \[ \sum_{X\in \mathcal{Q}_{j+1}} \mst_{t}(X)\le 2 \sum_{X\in \mathcal{Q}_j} \mst_{t}(X) \] since the minimum spanning tree of the drop-off locations of requests in a subgroup of $X\in  \mathcal{Q}_j$ have weight less than $\mst_{t}(X)$; 

Similarly, we have  \[ \sum_{X\in \mathcal{Q}_{j+2}} \mst_{s}(X)\le 2 \sum_{X\in \mathcal{Q}_{j+1}} \mst_{s}(X)\] and \[ \sum_{X\in \mathcal{Q}_{j+2}} \mst_{t}(X)\le \sum_{X\in \mathcal{Q}_{j+1}} \mst_{t}(X).\] 

Thus the Inequality~(\ref{ieq:OPT_bipar}) holds. It implies the following claim: 

\begin{claim}
\label{lemma:partitionalg_OPT_bipartition} 
For any group group $P$ and $l\ge 1$ with $|P|> 2^{l-1}$, suppose $h(P,l)-1<  \log \lceil \frac{|P|}{2^{l}}  \rceil \le h(P,l) $, following the \emph{Bi-partition rule}, for any even $ j\le h(P,l)$ we have:
\[ \sum_{X\in  \mathcal{Q}_{j}} \mst_{s,t}(X) \le 2^{\frac{j}{2}} \mst_{s,t} (P)  \] 
 for any odd $ j\le h(P,l)$ we have:
\[ \sum_{X\in  \mathcal{Q}_{j}} \mst_{s,t}(X) \le \frac{3}{2} \cdot  2^{\frac{j-1}{2}} \mst_{s,t} (P)  \] 
\end{claim} 

Then we prove Lemma~\ref{lemma:partitionalg_oneweight}: 
\begin{proof}  
Let  $\mathcal{P}^{*}_{>l}$ (resp, $ \mathcal{P}^{*}_{\le l}$) denote the union of group $P\in \mathcal{P}^{\star}$ with $|P|> 2^{l}$ (resp, $ |P|\le 2^{l}$).   We first partition all $P\in \mathcal{P}^{*}_{>l-1}$ following the \emph{Bi-partition rule} and obtain the $h(P,l)$-th partition $\mathcal{Q}_{h(P,l)}$ where $h(P,l)-1< \log \lfloor \frac{|P|}{2^l} \rfloor \le h(P,l)$.

Let $\mathcal{M}^*$ denote the union of $\mathcal{P}^{\star}_{\le l-1}$ and the partition of  item  $P\in \mathcal{P}^{\star}_{>l-1}$, i.e., $$\mathcal{M}^*= \bigcup_{P\in \mathcal{P}^{\star}_{\le l-1}} P \cup \bigcup_{P\in \mathcal{P}^{\star}_{>l-1}}  \mathcal{Q}_{h(P,l)}.$$ 
Notice that each item of  $P\in \mathcal{P}^{\star}_{\le l-1}$ contains not more than $ 2^{l-1}$ requests, each item of $\mathcal{Q}_{h(P,l)}$ with $P\in \mathcal{P}^{\star}_{>l-1}$ contains not more than $ 2^{l}$ requests and $\cup_{P\in \mathcal{M}^*} P=R $.

Consider a bipartite graph,  ($\mathfrak{M}_{l-1},\mathcal{M}^*, E$) where, for each request $r\in \bigcup_{\mathcal{Q}\in \mathfrak{M}_{l-1}} \bigcup_{X\in \mathcal{Q}} X$, there is an edge $e(r)$ between vertex $\mathcal{Q}\in \mathfrak{M}_{l-1}$ and $Y\in \mathcal{M}^*$ in which $r\in Y$.  
We claim that we can find a matching $M$ between $\mathfrak{M}_{l-1}$ and $\mathcal{M}^*$ that covers $\mathfrak{M}_{l-1}$ such that: for all items of each $ \mathcal{Q}_{h(P,l)} \subseteq \mathcal{M}^* $ where $P\in \mathcal{P}^{\star}_{>l-1}$, at most one item $Y\in \mathcal{Q}_{h(P,l)}$ has degree $1$ or $3$, all other items of $\mathcal{Q}_{h(P,l)}$ each has degree $0$ or $2$; for all $Y\in \mathcal{P}^{\star}_{\le l-1} \subseteq \mathcal{M}^*$, each has degree not more than $2$. The proof of this fact is summarized in Lemma~\ref{lemma:existmatching} below. 
 
We assign a value for each $\mathcal{Q}\in \mathfrak{M}_{l-1}$ based on the matching $M$: 
\begin{itemize}
    \item 
If only one $\mathcal{Q}\in \mathfrak{M}_{l-1}$ is matched with a $Y\in \mathcal{M}^*$,
$$c(\mathcal{Q})=\min_{r_i\in \bigcup_{X\in \mathcal{Q}} X} d(s_i,t_i) \le \max_{r_i\in Y} d(s_i, t_i);$$
\item If exactly two items $\mathcal{Q},\mathcal{Q}' \in \mathfrak{M}_{l-1}$ both are matched with a $Y\in \mathcal{M}^*$, let  $$c(\mathcal{Q})=c(\mathcal{Q}')=\frac{\mst_{s,t}(Y)}{2}.$$ 
\item Otherwise three items $\mathcal{Q},\mathcal{Q}',\mathcal{Q}'' \in \mathfrak{M}_{l-1}$  are matched with a $Y\in \mathcal{M}^*$, let  $$c(\mathcal{Q})=c(\mathcal{Q}')=\frac{\mst_{s,t}(Y)}{2},$$ 
$$c(\mathcal{Q}'')=\min_{r_i\in \bigcup_{X\in \mathcal{Q}''} X} d(s_i,t_i)  \le \max_{r_i\in Y} d(s_i, t_i).$$
\end{itemize}

Thus we have 
 \begin{align}  \sum_{\mathcal{Q}\in \mathfrak{M}_{l-1}} c(Q) &\le \sum_{Y\in \mathcal{P}^{\star}_{\le l-1}} \left( \mst_{s,t} (P)+  \max_{r_i\in Y} d(s_i,t_i) \right) \nonumber 
 \\ &+ 
\sum_{P\in \mathcal{P}^{\star}_{>l-1}} \left( \sum_{Y\in \mathcal{Q}_{h(P,l)}} \mst_{s,t} (Y) + \max_{r_i\in P} d(s_i,t_i)  \right) 
\label{ieq:OPT_sumvalue} 
\end{align}

Next we construct a matching $M^*_l$ for $\mathfrak{M}_{l-1}$ based on the matching $M$ and $\mathcal{P}^{\star}$:  any two set $\mathcal{Q} \in \mathfrak{M}_{l-1}$  with $c(\mathcal{Q})= \min_{r_i\in \bigcup_{X\in \mathcal{Q}}} d(s_i,t_i)$ are matched together, and we have
\begin{align*}
w(\mathcal{Q},\mathcal{Q}')=\min_{P\in \mathcal{Q}, P'\in \mathcal{Q}'}  \min \{w_1(P,P'),w_2(P,P')\}   \le \min_{P\in \mathcal{Q}, P'\in \mathcal{Q}'} w_2(P,P')  c(\mathcal{Q})+c(\mathcal{Q}').
\end{align*}
any two sets $\mathcal{Q},\mathcal{Q}'\in \mathcal{M}_{l-1}$ matched with  a same item of $\mathcal{Q}_{h(P,l)}$ are matched together, and we have 
\begin{align*}
w(\mathcal{Q},\mathcal{Q}')=\min_{P\in \mathcal{Q}, P'\in \mathcal{Q}'}  \min \{w_1(P,P'),w_2(P,P')\}   \le \min_{P\in \mathcal{Q}, P'\in \mathcal{Q}'} w_1(P,P')\le c(\mathcal{Q})+c(\mathcal{Q}').
\end{align*}

Then we have
\[w(M^*_l)\le  \sum_{\mathcal{Q}\in \mathfrak{M}_{l-1}} c(\mathcal{Q}).\]

According to the HG Algorithm~\ref{alg:partition_notexact} (Line 8), $M_l$ is the minimum weight matching, $$w(M_l)  \le w(M^*_l).$$

For $l$ with even $\log \lambda -l$, we have  
\begin{equation} \label{inq:ith_bound}
\begin{split}
 w(M_l) &\le w(M^*_l) \le \sum_{\mathcal{Q}\in \mathfrak{M}_{l-1}} c(\mathcal{Q})\\
 & \le \sum_{P\in \mathcal{P}^{\star}_{\le l-1}} \mst_{s,t} (P)+ \sum_{P\in \mathcal{P}^{\star}_{>l-1}}  \sum_{X\in \mathcal{Q}_{h(P,l)}} \mst_{s,t} (X)  +  \sum_{P\in \mathcal{P}^{\star}}  \max_{r_i\in P} d(s_i,t_i)\\
& \le  \sum_{P\in \mathcal{P}^{\star}} 2^{\frac{\log \lambda-l }{2}}\mst_{s,t}(P)+  \sum_{P\in \mathcal{P}^{\star}}  \max_{r_i\in P} d(s_i,t_i)
\end{split}
\end{equation}
The third inequality follows from Inequality~(\ref{ieq:OPT_sumvalue});  
The last inequality follows from Claim~\ref{lemma:partitionalg_OPT_bipartition} and the fact that $ h(P,l)\le \log \lambda -l$  since $|P|\le \lambda$.

Similarly, we get the proof for  $l$ with odd $\log \lambda -l$.
\end{proof}

There is one piece remaining to be filled. In the proof of Lemma~\ref{lemma:partitionalg_oneweight} we made a structural claim on the existence of certain matching $M$ between $\mathfrak{M}_{l-1}$ and $\mathcal{M}^*$. Now we formally prove it below.
\begin{restatable}{lemma}{LemmaExistMatching}
\label{lemma:existmatching} 
Consider a cluster $\mathcal{P}$, and a set of groups $\mathcal{M}^*= (\bigcup_{\mathcal{Q}} \mathcal{Q}) \cup \mathcal{M}^*_{<}$ which satisfies:
\begin{itemize}
   \item Each $P\in\mathcal{P}$ contains exactly $m$ requests and $P\cap P' =\emptyset$ for any $P,P'\in \mathcal{P}$; 
    \item Each $ P\in  \mathcal{M}^*_{<}$  contains no more than $m$ requests;
    \item For each $\mathcal{Q} $, at most one group $ P\in  \mathcal{Q}$  contains no more than $2m$ requests and all other $P\in \mathcal{Q}$ contains exactly $2m$ requests;
    \item $P\cap P' =\emptyset$ for any $P,P'\in \mathcal{M}^*$.  
\end{itemize}  
If $\bigcup_{P\in \mathcal{P}} P =(\bigcup_{\mathcal{Q}}  \bigcup_{P\in \mathcal{Q}} P) \cup (\bigcup_{P\in \mathcal{M}^*_{<}} P)$, then   
we can find a matching $M$ between $\mathcal{P}$ and $\mathcal{M}^*$ that covers all of $\mathcal{P}$ such that: 
\begin{enumerate}[(1)]
    \item For all items $P\in \mathcal{Q}  \subseteq \mathcal{M}^*$, at most one item has degree $1$ or $3$, all other items each has degree $0$ or $2$; 
    
    \item For each $P \in \mathcal{M}^*_{<} \subseteq \mathcal{M}^*$, $P$ has degree no more than $2$. 
\end{enumerate}
\end{restatable}
\begin{proof}
The matching $M$ is constructed in three steps, as shown in Fig.~\ref{fig:exist_matching_bigsmall}:

\begin{figure}[htbp] 
	\centerline{\includegraphics[width=0.8\linewidth]{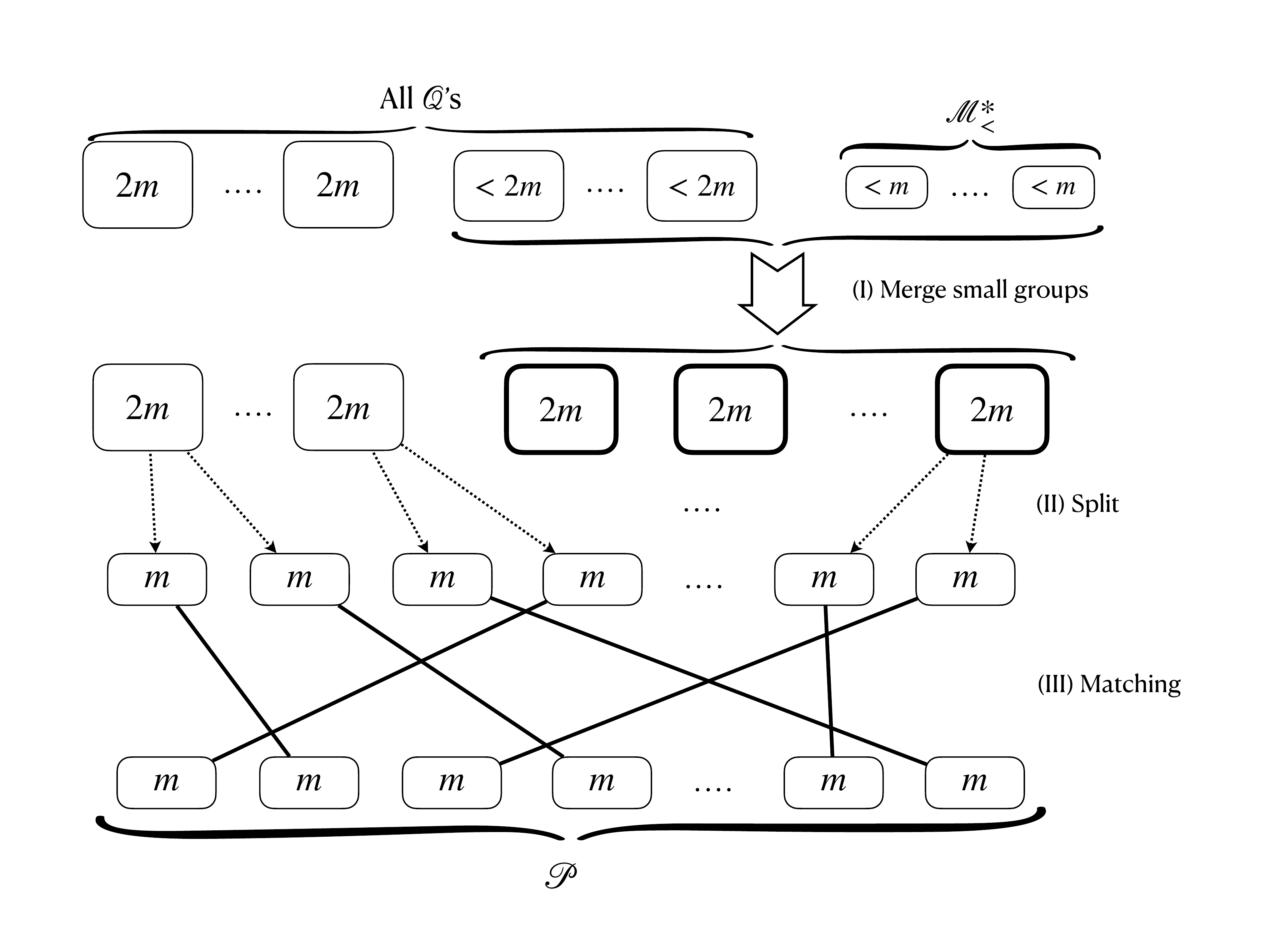}} 
	\caption{Matching between $\calP$ and $\mathcal{M^*}$}
	\label{fig:exist_matching_bigsmall} 
\end{figure}

\noindent\textbf{(I). Merge request groups $\bigcup_{P\in \mathcal{M}^*, |P|<2m} P$ into groups $\mathcal{J}$ such that $|P|=2m$ for each $P\in \mathcal{J}$:}   
In this step, we layout the requests in these groups consecutively in a row, and then split the merged requests into groups in a fixed order where each group contains exactly $2m$ requests.   

\noindent\textbf{(II). Split the request groups $\mathcal{J}$ into groups $\mathcal{I}$  such that $|P|=m$ for each $P\in \mathcal{I}$.}  
In this step, we split each request group $P\in \mathcal{J}$ into two groups such that each group contains exactly $m$ requests in a fixed order.  
Observe that for each group $P\in \mathcal{M}^*$ such that $m<|P|<2m$, the requests in $P$ are present in at most three groups of $\mathcal{I}$; for each group $P\in \mathcal{M}^*$ such that $|P|<m$, the requests in $P$ are present in at most two groups of $\mathcal{I}$.

\noindent\textbf{(III). Find a matching between $\mathcal{P}$ and $\mathcal{I}$.} 
After the above two steps, we know that for any $P\in \mathcal{P}$ and $Q\in \mathcal{I}$, $|P|=|Q|=m$ and $\bigcup_{P\in \mathcal{P}}=\bigcup_{Q\in \mathcal{I}} Q$. 
Consider a bipartite, $m$-regular, multi-graph $(\mathcal{P}, \mathcal{I}, E)$ where, for each request $r$, there is an edge $e(r)$ between vertex  $P\in \mathcal{P}$ and $Q\in  \mathcal{I}$, whenever $r\in  P\cap Q$.  
Hall's marriage theorem is known to give a necessary and sufficient condition for finding a perfect matching that covers one side of the graph $G=(X\cup Y, E)$: there is an $X$-perfect matching if and only if $ |W|\leq |N_G(W)| $ for every subset $W$ of $X$ where $N_G(W)$ denote the neighborhood of $W$ in $ G$.     
We can find a perfect matching $M'$ among  $\mathcal{P}$ and $\mathcal{I}$. Let $E(M')$ denote the edges (also, the corresponding requests) in matching $M'$.

Based on the matching $M'$ between $\mathcal{P}$ and $\mathcal{I}$, we can construct a matching $M$ between $\mathcal{P}$ and $\mathcal{M}^*$: for requests (or edges) in $E(M')$ which are present in a same group of $\mathcal{M}^*$, we match this group to the groups in $\mathcal{P}$ connected by these edges.  
The matching $M$ has the following properties: (a) For all items $P\in \mathcal{M}^*$ with $|P|=2m$, all items each has degree $2$;  
(b) For each $P\in \mathcal{M}^*$ with $m<|P|<2m$, $P$ has degree no more than $3$ because the requests in $P$ are present in at most three groups of $\mathcal{I}$; According to the definition, in each $\mathcal{Q}  \subseteq \mathcal{M}^* $, at most one group $ P\in  \mathcal{Q}$ contains no more than $2m$ requests, thus at most one group in each  $\mathcal{Q} $ are matched with no more than three groups of $\mathcal{I}$, implying that the degree these matched request' corresponding request group in $\mathcal{M}^*$ has degree not more than $3$; 
(c) For each $P \in \mathcal{M}^*$ with $|P|\le m$, $P$ has degree no more than $2$ because the requests in $P$ are present in at most two groups of $\mathcal{I}$.       
 That means $M$ satisfies the condition (1) and (2).    
\end{proof}

Based on Lemma~\ref{lemma:partitionalg_oneweight}, it is easy to obtain the following lemma:
\begin{restatable}[Total Grouping Cost]{lemma}{TotalGroupingCost}
\label{lemma:partitionalg_sumweight}  
 The total grouping 
cost 
\begin{align*}
\sum_{l=1}^{\log \lambda} w(M_l)  \le & O(\sqrt{\lambda})  \sum_{P\in \mathcal{P}^{\star}} \mst_{s,t}(P)+
 \log \lambda  \cdot  \sum_{P\in \mathcal{P}^{\star}}  \max_{r_i\in P} d(s_i,t_i). 
\end{align*}
\end{restatable}

Now we are ready to bound the cost of the capacity-bounded groups:

\begin{restatable}{theorem}{ThmPartition}
\label{thm:partition}
The \textsc{Hierarchical Grouping} algorithm  
runs in polynomial time and outputs a feasible partition $\calP$ of $R$, such that (1) $\forall\,P\in\calP, |P|\leq\lambda$; (2) $\cost(\mathcal{P})\leq O(\sqrt{\lambda})\cdot \cost(\calP^*)$, where $\calP^*$ is the partition of an optimal solution for the capacitated grouping problem; and (3) $\forall\,P\in\calP$ one can efficiently find a feasible walk $\bfw_P$ traversing $P$, 
such that $\sum_{P\in\calP}\cost(\bfw_P)\leq O(\cost(\calP))$.
\end{restatable}   
\begin{proof} 
First notice that, given any group of request $P$, we have the following lowerbounds for the cost of the optimal walk $$\bfw^*_P:=\argmin_{\text{feasible } \bfw_P}\cost(\bfw_P): \mst_{s,t}(P)\leq 2\cost(\bfw^*), $$  $$\min_{r_i\in P} d(s_i,t_i)\leq \max_{r_i\in P} d(s_i,t_i) \le  \cost(\bfw_P^*).$$

Then by Lemma~\ref{lemma:partitionalg_each}, Lemma~\ref{lemma:partitionalg_lastweight}, and Lemma~\ref{lemma:partitionalg_sumweight},
we have 
\begin{align*}
& \sum_{P\in\mathcal{P}} (\mst_{s,t}(P)+\min_{r_i\in P} d(s_i,t_i)) \\
\le& O(\sqrt{\lambda})\cdot \sum_{P\in \mathcal{P}^{\star}} \mst_{s,t}(P) + (1+\log \lambda) \cdot \max_{r_i\in P} d(s_i,t_i)\\
\le& O(\sqrt{\lambda})\cdot \sum_{P\in \mathcal{P}^{\star}} \cost(\bfw^*_P) \\
= & O(\sqrt{\lambda})\cdot \cost(\calP^*).
\end{align*}
In particular, the two \emph{minimum spanning trees} that obtains $\mst_{s,t}(P)$ gives a feasible walk for serving requests in $P$ as follows: Let $(s,t)=\argmin_{(s_i,t_i)\in P}d(s_i,t_i)$. Pick any $s'\in \{s_r:r\in P\}, s'\neq s$, and by the classical Christofide's algorithm we easily find a $s'$-$s$ TSP path on $\{s_r:r\in P\}$ with cost at most $O(\mst_s(P))$. Similarly, we can find a $t$-$t'$ TSP path on $\{t_r: r\in P\}$ with cost at most $O(\mst_t(P))$. By gluing the two TSP path together through $(s,t)$ we get a feasible walk, $\bfw_P$, on $P$ with cost $O(\mst_{s,t}(P)+\min_{r_i\in P} d(s_i,t_i))$. This is the $\bfw_P$ used in the theorem statement.  
Since $\cost(\calP)\leq \sum_{P\in\calP}\cost(\bfw_P)$, the existence of $\bfw_P$ also implies $\cost(\calP)\leq O(\sqrt{\lambda})\cdot\cost(\calP^*)$.
\end{proof}

%% file: routing.tex
\section{Part (II): Routing} 
\label{sec:assignandroute}

After invoking Algorithm~\ref{alg:partition_notexact} to get the partition $\calP$ of requests, we now describe how to find actual routes for the vehicles --- the \emph{assignment} $\calA$. 
The requests will be served group-by-group, meaning that each group is served \emph{exclusively} and \emph{non-preemptively}\footnote{By \emph{non-preemptive}, we mean a vehicle must finish serving all requests of a group before it can start serving other groups.} by some vehicle. Such route is of course unlikely to be optimal, but the previous hierarchical grouping phase provides a good structure, which allows us to prove a good approximation ratio.

Our routing phase consists of two steps: First we will generate a set of graphs, specifically we call a \emph{rooted spanning forest} (defined below) $\calF$ that connects each group to exactly one vehicle in $V_K$; Then we design a routing plan that schedules the vehicles to serve its connected groups along the edges of $\calF$.
We now describe the algorithm formally.

\vspace{0.5\baselineskip}
\noindent\textbf{(I). Finding the rooted spanning forest $\calF$.} 
First let us formally define the rooted spanning forest.
\begin{definition}[Rooted Spanning Forest (RSF)]\label{def:rooted-spanning-forest}
Given a weighted graph $G=(V,E)$ with edge cost $c:E\mapsto\bbR_{\geq0}$ and a root set $U\subseteq V$, we say a set $\calF=\{T_i\}_i$ of (disjoint) trees is a \emph{rooted spanning forest (RSF)}, if 
\begin{enumerate}
    \item Each $T_i\in\calF$ is a tree rooted at some vertex of $U$;

    \item $T_i\cap T_j=\emptyset$ for any $i\neq j$;

    \item For any non-root $v\in V\setminus U$, there is some $T_i\in\calF$ contains $v$.
\end{enumerate}

\noindent Lastly, define the cost of $\calF$ as $c(\calF)=\sum_{T\in\calF}c(T)=\sum_{T\in\calF}\sum_{e\in T}c(e)$. We say $\calF$ is a \emph{Minimum Rooted Spanning Forest (MRSF)} if $\calF$ achieves minimum cost among all rooted spanning forests.
\end{definition}

\begin{claim}\label{clm:mrsf}
Given $G, c, U$ as above in Definition~\ref{def:rooted-spanning-forest}, we can find a minimum rooted spanning forest (MRSF) $\calF$ in polynomial time.
\end{claim}
\begin{proof}
We first contract $U$ to a single vertex $u_0$. Then for any non-root vertex $v$, we merge all parallel edges between $v$ and $u_0$, and re-define the cost $c(u_0, v):=\min_{x\in U}c(x, v)$. Then we find a minimum spanning tree $T$ in this contracted graph (using, e.g., Prim's algorithm). Now we un-contract $u_0$ back to $U$: if a non-root vertex $v$ was connected to $u_0$ by $T$, then after un-contraction it is connected to its closest neighbor in $U$. 

It is easy to see that after un-contraction $T$ becomes a rooted spanning forest (with the same cost) in the original graph $G$, which is our desired solution $\calF$. To see that $\calF$ is minimum: suppose for contradiction there is another rooted spanning forest $\calF'$ with smaller cost, then if we contract $U$, $\calF'$ gives a spanning tree cheaper than $T$, contradicting with $T$ being minimum.
\end{proof} 
We will build a minimum rooted spanning forest for the groups $\calP$ with root set $V_K$ (the vehicles). Formally, consider the complete graph over vertices $V_K\cup \calP$, with $V_K$ being the root set. We define the edge cost $c$ on this graph as follows: recall $d$ is the underlying metric, let $c$ be
\begin{equation}\label{eq:rsf-cost}
\left.
\begin{aligned}
  c(P, P') &:= \min_{r_i\in P, r_j\in P'}d(s_i, s_j), & P, P'\in\calP\\
  c(u, P) &:= \min_{r_i\in P} d(u, s_i), & u\in V_K, P\in\calP\\
  c(u, v) &:= d(u,v), & u, v\in V_K \\
\end{aligned}
\quad\right\}
\end{equation}
Now, using Claim~\ref{clm:mrsf} we find a minimum rooted spanning forest $\calF$ w.r.t. the cost $c$ above. Note that each tree $\calT\in\calF$ contains \emph{exactly} one vertex from $V_K$, which we will designate as the \emph{root} of $\calT$. 

The above process is summarized in Procedure~\ref{proc:mrsf}.  
 
{
  
\begin{algorithm}
 \floatname{algorithm}{Procedure}
\begin{algorithmic}[1]
\small
 \caption{\small{\textsc{MRSF}($V_K, \calP$)} 	\label{proc:mrsf}}

\REQUIRE Vehicle locations $V_K$ and partition $\calP$
\ENSURE A rooted spanning forest $\calF$ on $V_K\cup\calP$.

\STATE Let $c$ be given as in Eq~\eqref{eq:rsf-cost}
\STATE Find a minimum rooted spanning forest $\calF$ on $V_K\cup S$ with cost function $d$, based on Claim~\ref{clm:mrsf}.

\RETURN $\calF$
\end{algorithmic}
\end{algorithm}
}

\vspace{0.5\baselineskip}
\noindent\textbf{(II). DFS on $\calF$ to serve all requests.} For each tree $\calT\in\calF$, all of its request will be served using only the unique vehicle that is located at the root of $\calT$. So, now we can focus on serving a single tree $\calT$. Roughly speaking, the vehicle leaves the root of $\calT$ and serves each group in a \emph{depth-first} manner along the edges of $\calT$. But since each group contains multiple locations, the apparent question is, how exactly does a vehicle move?

Let $S(P)$ denote all the pickup locations from group $P$. Recall that, by construction each edge $(P,P')$ (or $(p_k,P), p_k\in V_K$) in $\calT$ uniquely corresponds to an edge $(s,s')$ (resp. $(p_k,s)$) for some $s\in S(P)$ and $s'\in S(P')$, so we can think of $P, P'$ are connected via the ``portals'' $s, s'$. We also denote the portal via which a group $P$ connects to its parent as $s_0(P)$. For ease of presentation, we also define $s_0(p_k):=p_k$ for all $p_k\in V_K$. The vehicle will always enter $P$ at $s_0(P)$ from its parent.

Then, let $\bfw_P$ be the walk serving $P$ that is guaranteed by Theorem~\ref{thm:partition}, and let $s_1$ be the starting point of $\bfw_P$. The vehicle will first move to $s_1$ from $s_0(P)$ and serve all requests of $P$ by following $\bfw_P$, then traverse $S(P)$ again in the order determined by $\bfw_P$, serving its children groups recursively. Finally, the vehicle move back to $s_0(P)$ and return to $P$'s parent. The process is summarized in Procedure~\ref{proc:dfs}.
We also provide an example in Figure~\ref{fig:DFS}. 

{
\floatname{algorithm}{Procedure}
\begin{algorithm}

\caption{\small{\textsc{Dfs}($P$, $\calT$)} 	\label{proc:dfs}}
\small
\begin{algorithmic}[1]
 
\REQUIRE $P\in V_K\cup \calP$, and $\calT$ is a tree on $V_K\cup\calP$ containing $P$
\ENSURE A feasible walk $\bfw$ covering the subtree of $\calT$ rooted at $P$
   
\STATE $\bfw\gets \langle s_0(P)\rangle$
\IF{$P\in\calP$} 
\STATE $\bfw_P \gets$the walk guaranteed by Theorem~\ref{thm:partition}
\ELSE 
\STATE $\bfw_P\gets\langle P \rangle$ \COMMENT{If $P\in V_K$ is a vehicle location}
\ENDIF
\STATE $s\gets$starting point of $\bfw_P$
\STATE Append $\bfw_P$ to $\bfw$ \COMMENT{Serve $P$}

\FOR{each $s\in S(P)$ in the order of $\bfw_P$}
  \FOR{each child group $P'$ of $P$ connected via $s$} 
    \STATE Append $s$ to $\bfw$
    \STATE $\bfw'\gets$\textsc{Dfs}($P'$, $\calT$)
    \STATE Append $\bfw'$ to $\bfw$
  \ENDFOR   
  \STATE Append $s$ to $\bfw$
\ENDFOR   
\STATE Append $s_0(P)$ to $\bfw$  
\RETURN $\bfw$
\end{algorithmic}
\end{algorithm}
} 
\begin{algorithm}
\small
\begin{algorithmic}[1]
 \caption{\small{\textsc{Routing}}($R, \calP, V_K$)}
 	\label{alg:assignandroute}
\REQUIRE $R$, $\calP$ and $V_K$
\ENSURE An assignment $\mathcal{A}$ that serves $R$
   
\STATE $S\gets$ set of all pick-up locations of $R$.
\STATE $\calF\gets\textsc{Mrsf}(V_K, \calP)$ \COMMENT{Proc.~\ref{proc:mrsf}}
\STATE $\calA\gets\emptyset$
\FOR{each tree $\calT\in\calF$}
\STATE $p_k\gets$root of $\calT$
\STATE $\bfw_\calT\gets\textsc{Dfs}(p_k,\calT)$ \COMMENT{Proc.~\ref{proc:dfs}}
\ENDFOR 

\RETURN $\mathcal{A}=\{\bfw_\calT: \calT\in\calF\}$
\end{algorithmic}
\end{algorithm}

\begin{figure*}[!htbp]
    \centering
    \includegraphics[width=0.9\textwidth]{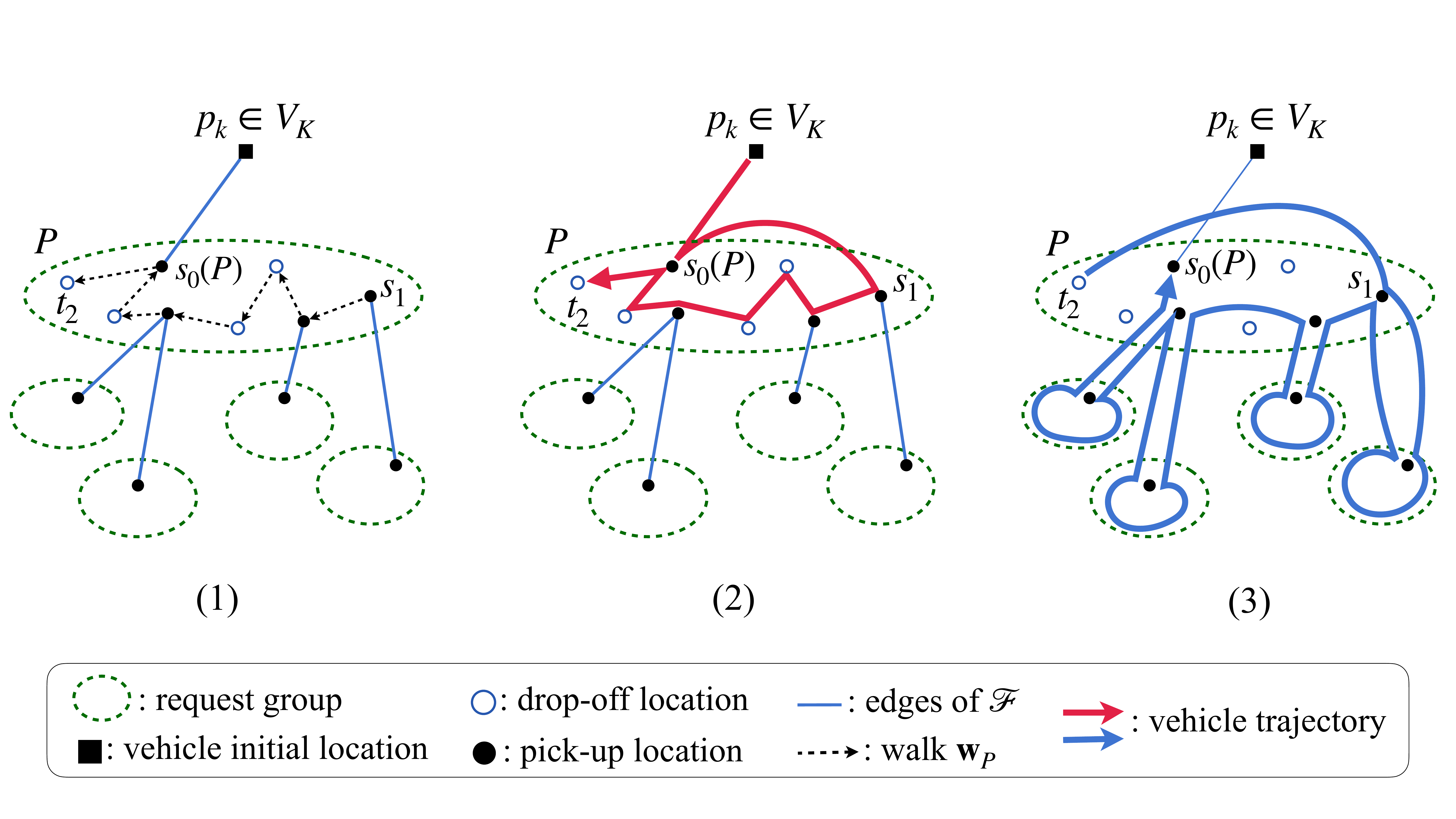}
    \caption{Example of serving the groups using DFS. Figure~ (1) shows the constructed spanning forest $\calF$, and $\bfw_P$ that starts at some $s_1$ and ends at $t_2$. Figure~(2) shows how the vehicle serves $P$: it enters $P$ at $s_0(P)$, then move to $s_1$ and serves $P$ by following $\bfw_P$. Figure~(3) shows how the vehicle recursively serve $P$'s children after serving $P$: it first moves back to $s_1$, and visit the children of $P$ in the order of $\bfw_P$ (breaking ties arbitrarily).  \label{fig:DFS}}
\end{figure*}

With Procedure~\ref{proc:mrsf} (\textsc{Mrsf}) and \ref{proc:dfs} (\textsc{Dfs}) at hand, the final \textsc{Routing} algorithm is quite straightforward: just apply the \textsc{Dfs} to each tree of the rooted spanning forest returned by \textsc{Mrsf}.

\section{Proof of Theorem~\ref{thm:approx-darp}}
Now we can prove our main theorem, see Theorem~\ref{thm:approx-darp}. restated below. \ThmApproxDaRP* 
Let $\calA$ be the solution returned by HGR (Algorithm~\ref{alg:main-algorithm}), $\calP$ be the partition output by \textsc{Hierarchical Grouping} (Algorithm~\ref{alg:partition_notexact}), and $\calF$ be the rooted spanning forest obtainted by \textsc{Mrsf}$(V_K, \calP)$ (Procedure~\ref{proc:mrsf}). 
First, we have the following simple claim on the cost of $\calF$.
\begin{restatable}{claim}{ClmCostBwGroups}
\label{clm:cost-between-groups} 
Let $S$ denote the set of all pick-up locations in $R$. For any request partition $\calP$, let $\calF=\textsc{Mrsf}(V_K,\calP)$ (Procedure~\ref{proc:mrsf}) and $\cost(\calF):=\sum_{e\in \calF}d(e)$, then we have:
$$\cost(\calF)\le \cost(\calA^*), $$ 
where $\calA^*$ is any feasible solution to the original Dial-a-Ride problem.
\end{restatable}

\begin{proof}
The solution $\calA^*$, being a collection of walks, can also be thought as a rooted spanning forest over all the pick-up (as well as the drop-off) locations: View each walk as a sequence of weighted edges (with weight given by metric $d$), then every $s\in S$ is connected to some vehicle from $V_K$ as a root. Now we modify $\calA^*$ in 3 steps to make it also a RSF on $V_K\cup\calP$ using the same cost function $c$ (Eq~\eqref{eq:rsf-cost}): 
\begin{enumerate}
    \item Shortcut every drop-off locations to make $\calA^*$ a rooted spanning forest on $S$ only. Since the distance $d$ is a metric, this only reduces $\calA^*$'s cost; 
    \item Then we contract each group of $\calP$ to a single vertex, which preserves only between-group edges of $\calA^*$. This again only reduces its cost;
    \item Finally, since each remaining edge is either between two groups or between a group and a vertex in $V_K$, we can reassign its cost using $c$. By definition of $c$, this only reduces the total cost.
\end{enumerate}
Denote the resulting graph as $\calA'$. By construction, it is a rooted spanning subgraph over $V_K\cup\calP$ with at most the same cost of $\calA^*$. Then by definition we have $\cost(\calF)\leq\cost(\calA')\leq\cost(\calA^*)$. 
\end{proof} 

Now we give the proof for the main theorem.
  
\begin{proof}[Proof of Theorem~\ref{thm:approx-darp}]
Consider any optimal solution $\calA'$, which is also a collection of walks starting from locations in $V_K$. Applying Fact~\ref{fact:construction} to each walk of $\calA'$, we get a new solution $\calA^*$ such that (1) every walk of $\calA^*$ serves requests in a group-by-group manner, where each group is of size at most $\lambda$; and (2) $\cost(\calA^*)\leq\cost(\calA')\cdot O(\log n)=O(\log n)\OPT$.

Then we show $\cost(\calA)$ can be bounded by $O(\sqrt{\lambda})\cdot \cost(\calA^*)$, which will give $\cost(\calA)\leq O(\sqrt{\lambda}\log n)\cdot\OPT$.
$\cost(\calA)$ can be decomposed into two parts: the cost of traveling along edges in $\calF$, and the cost of moving within each group. By the nature of DFS, each edge of $\calF$ is traversed exactly twice, therefore the first part of the cost is at most $2\cost(\calF)\leq2\cost(\calA^*)$ by Claim~\ref{clm:cost-between-groups}.

For the second part of cost, we fix a request group $P$ and consider the total travel distance of the vehicle within $P$. Recall $s_0(P)$ is the ``portal'' connecting $P$ with its parent, and $\bfw_P$ is the walk given by Theorem~\ref{thm:partition} that serves all requests in $P$. Let $s'$ be the starting location of $\bfw_P$. The vehicle first moves from $s_0(P)$ to $s'$, which takes at most $\cost(\bfw_P)$. It then serves all of $P$ by following $\bfw_P$, which takes another $\cost(\bfw_P)$. The vehicle then moves back to $s'$ and traverse $\bfw_P$ again to recursively serve all children groups of $P$, and finally moves to $s_0(P)$. This process will cost at most another $3\cost(\bfw_P)$. So overall the travel distance within $P$ is at most $5\cost(\bfw_P)$. (Note this is apparently not the most efficient moving strategy, but it suffices to give the desired bound)

To summarize, 
$\cost(\calA)\leq 2\cost(\calA^*)+5\sum_{P\in\calP}\cost(\bfw_P)\leq O(\sqrt{\lambda})\cdot\cost(\calA^*)$,
where the last inequality is by Theorem~\ref{thm:partition}. This concludes our proof.
\end{proof} 

%% file: experiments.tex
\section{Computational Experiments}\label{sec:experiments}

\subsection{DaRP Instances and Baselines}

\noindent\textbf{Synthetic datasets.}
We benchmark HGR on two synthetic datasets. In the first dataset (SY-U), locations (i.e., request pickups, drop-offs and drivers' initial locations) are randomly generated from a uniform distribution on a $[0,100]^2$ grid. In the second dataset (SY-G), locations are randomly generated from a mixed-Gaussian distribution with $Z$ clusters. Each cluster corresponds to a bivariate Gaussian distribution whose center is drawn from a uniform distribution on a $[0,1000]^2$ grid and with covariance matrix given by $\sigma^2\mathbf{I}$. For the SY-G distribution, several combinations of parameters $Z$ and $\sigma$ are tested, as detailed in Table~\ref{tbl:param-setting}.


\vspace{0.5\baselineskip}
\noindent\textbf{Realworld datasets.}
We also test the algorithms on two realworld datasets consisting of transportation data from New York City (NYC) and San Francisco (SFO). 
\begin{itemize}
    \item NYC: We use NYC Taxi \& Limousine Commission Trip Record Data \cite{NYC-TLC-trip-rec}. In particular, we randomly select 10,000 trip records from May/2016.
    \item SFO: We use the Cab Spotting Data \cite{SFO-cab-rec}, which records roughly 500 taxis' trace data in a period of 30 days. Again, we randomly select 10,000 trip records from the original dataset.
\end{itemize}

In these datasets a location is specified by its latitude and longitude coordinates. The distance between two locations is defined to be the graph (shortest-path) distance calculated via the actual road map of the two cities, which we obtain from the OpenStreetMap \cite{open-street-map}. All parameter settings are detailed in Table~\ref{tbl:param-setting}.
 
\begin{table}[!htbp]
\small
\centering
\subfloat{
\begin{tabular}{ccc} 
 \toprule
 & SY-U, NYC and SFO & SY-G\\
 \midrule
 $n$ & 2, 4, 6, \textbf{8}, 10 $(\times10^3)$& 2, 4, \textbf{6} $(\times10^3)$\\
 $m$ &  30, 60, \textbf{90}, 120, 150 & 30, 60, \textbf{90}\\
 $\lambda$ & 2, 4, 8, 16, \textbf{32}, 64 & 4, 8, \textbf{16}\\
 $Z$ && 5, \textbf{10}, 20, 50, 100 \\
 $\sigma$ && 5, 10, 25, \textbf{50}, 100, 250 \\
 \bottomrule
\end{tabular}
}
\\
\caption{\label{tbl:param-setting} Parameter settings for the datasets. Bold values indicate fixed parameter values for the sensitivity analyses.}
\end{table}

\noindent\textbf{Metrics and Baselines.} We measure the performance of the algorithms with respect to two objectives. The first is the \emph{total travel distance}, which is what \HRA is set to optimize. The second objective is the \emph{total in-transit latency}: the in-transit latency for a rider is the time length he/she is on board, i.e., the time between being picked-up and finally dropped-off at his/her destination. Total in-transit latency corresponds to the sum of in-transit latency over all rides. This objective is crucial for customer experience as most riders prefer to reach their destination as quickly as possible after being picked-up.

For baselines, since the Dial-a-Ride problem is a classical problem in operations research, there are too many proposed methods for an exhaustive comparison. However, many popular heuristic or exact algorithms are aiming at solving moderate-sized instances with at most 1000 requests (see e.g.\cite{ho2018survey} for a recent comprehensive benchmark\footnote{ The benchmark results are also available at \url{https://sites.google.com/site/darpsurvey/instances}}), while our motivation is to solve, quickly and with a \emph{worst-case theoretical guarantee}, large-scale dial-a-ride problems that appear in online ride-sharing applications. Therefore, we select the following two baselines in the experiment:
\begin{itemize}
    \item \PGP\cite{TongZZCYX18-GDP}: This is a heuristic algorithm that builds routes incrementally by greedily inserting new requests. It is designed to optimize the total travel distance and can be implemented very efficiently, though no approximation guarantee is known.
    \item \FESI\cite{ZengTC19-LMD}: This algorithm optimizes the \emph{makespan} of vehicles, i.e., the maximum distance traveled by the vehicles. But as claimed in their paper, \FESI also obtains small total travel distance on various datasets, often comparable with \PGP. We remark that \FESI has a $O(\sqrt{\lambda}\log n)$ approximation guarantee in terms of the makespan objective, but no guarantee for total travel distance is known.
\end{itemize}
The two baselines aim at a very similar or exactly the same application scenario. Both baselines (as well as our algorithm) are able to handle thousands of requests efficiently. Besides, both \PGP and \FESI report comparisons with many popular OR algorithms (including the ALNS~\cite{gschwind2019ALNS} algorithm which gives top performance in the benchmark \cite{ho2018survey} mentioned above), and generally exhibit superior performance in most large instances. Thus, we believe the two methods are suitable representatives of the state-of-the-art.

\vspace{0.5\baselineskip}
\noindent\textbf{Implementation.}
We use the publicly-available code provided by \cite{ZengTC19-LMD} for \FESI. Other algorithms (including ours) are implemented in C++. All experiments are conducted on a single core of an Intel\textsuperscript{\textregistered} Xeon\textsuperscript{\textregistered} Gold 6130 (2.1GHz) processor with 32GB of available RAM. As \FESI is a randomized algorithm, we run \FESI 10 times on each instance and report the average results. Our implementations, with which all results can be reproduced, is publicly available at {\emph{https://github.com/amflorio/hga-dial-a-ride}.

\subsection{Computational Results}

Figure~\ref{fig:synthetic-Uniform} to \ref{fig:synthetic-GMM-2} depict the result on synthetic data sets. Results on realworld datasets are shown in Figure~\ref{fig:nyc-nonEuclidean} and \ref{fig:sfo-nonEuclidean}. Overall, our algorithm exhibits clear superiority on both objectives in almost all parameter regimes and datasets. Now, we discuss the effect of each parameter in more details.

\vspace{0.5\baselineskip}
\noindent\textbf{Synthetic datasets.} 
Figure~\ref{fig:synthetic-Uniform} and \ref{fig:synthetic-GMM-1} illustrates the effect of $n,m,\lambda$ on the algorithms' performance, for data generated from uniform and GMM distribution, respectively. We remark that the trend on the synthetic datasets are much similar to that on the realworld datasets (Figure~\ref{fig:nyc-nonEuclidean} and \ref{fig:sfo-nonEuclidean}), therefore we postpone a detailed discussion on the effect of $n,m,\lambda$ to the later part when reporting results on realworld datasets. Generally speaking, for the total travel distance objective, our algorithm \HRA consistently performs the best, while \FESI is the worst. When it comes to in-transit latency, our algorithm still achieves the best objective in most parameter regime, but \FESI is able to exploit more or larger vehicles, and is likely to provide better latency.

\begin{figure}[!htbp]
\captionsetup[subfigure]{labelformat=empty}
   \centering
   \subfloat
   {
     \includegraphics[width=0.9\textwidth]{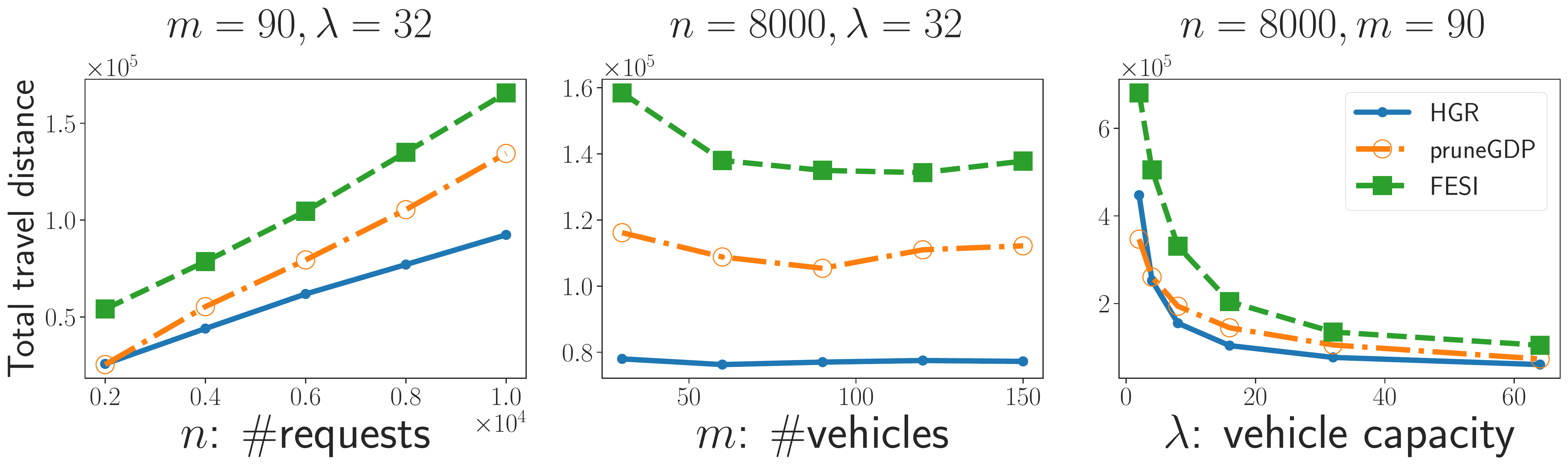}
   }
   \\
   \subfloat
   {
     \includegraphics[width=0.9\textwidth]{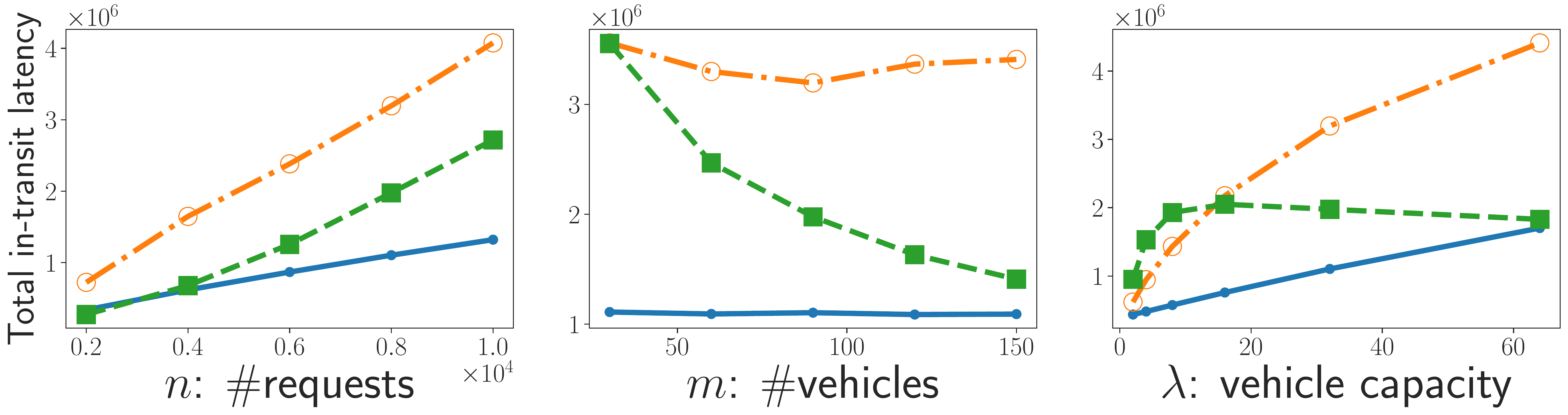}
   }
   \caption{Results on the Uniform synthetic dataset (SY-U) with varying $n, m, \lambda$.}
   \label{fig:synthetic-Uniform}
\end{figure}

\begin{figure}[!htbp]
\captionsetup[subfigure]{labelformat=empty}
   \centering
   \subfloat
   {
     \includegraphics[width=0.9\textwidth]{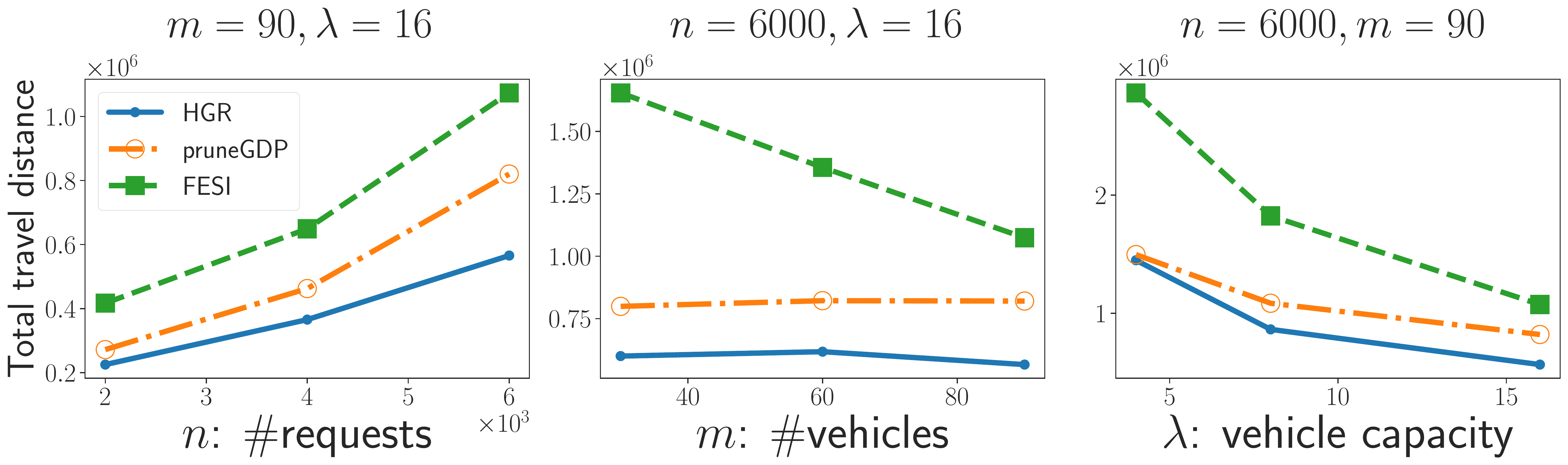}
   }
   \\
   \subfloat
   {
     \includegraphics[width=0.9\textwidth]{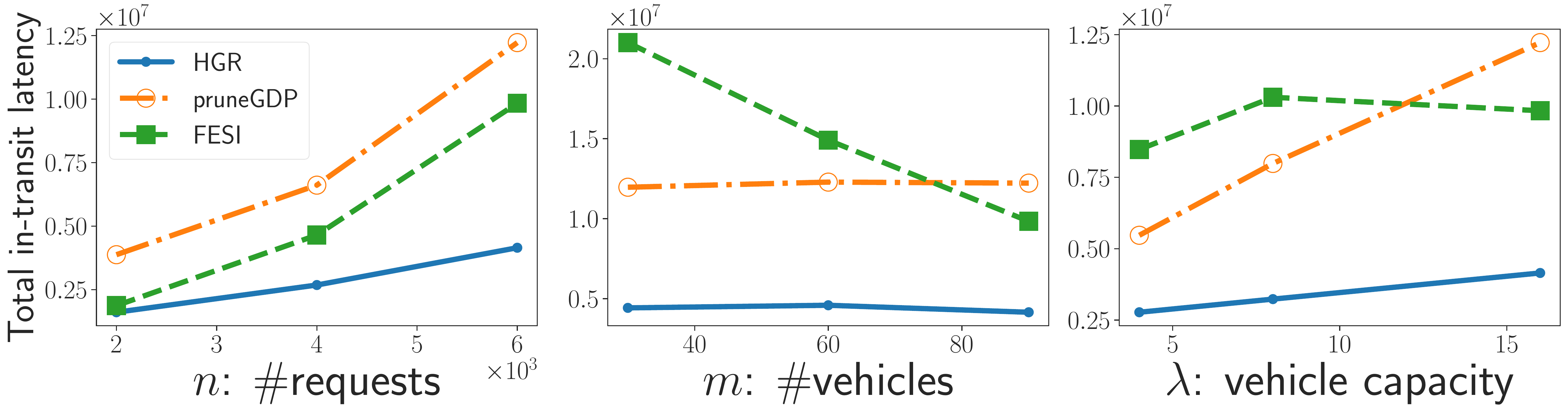}
   }
   \caption{Results on the GMM synthetic dataset (I): varying $n, m, \lambda$.}
   \label{fig:synthetic-GMM-1}
\end{figure}

\begin{figure}[!htbp]
    \centering
    \includegraphics[width=0.7\textwidth]{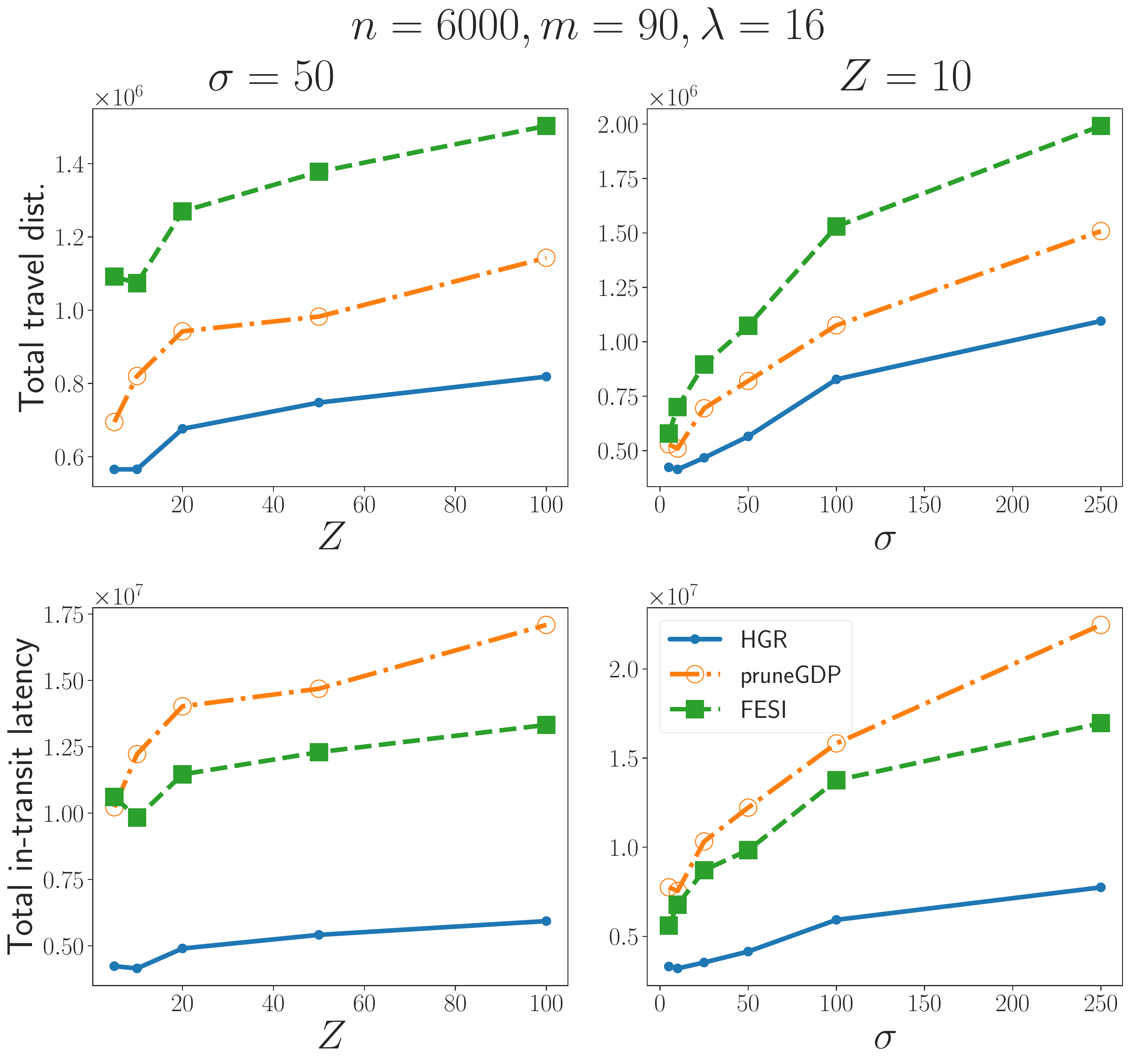}
    \caption{Results on the GMM synthetic dataset (II): varying $Z$ and $\sigma$.\label{fig:synthetic-GMM-2}}
\end{figure}
Figure~\ref{fig:synthetic-GMM-2} shows how the distribution parameters (specifically, the number of clusters and covariance of the GMM) affects algorithm performance
The first two columns of Figure~\ref{fig:synthetic-GMM-2} show the effect of varying the number of clusters ($Z$) of the GMM, and the remaining two show the effect of varying $\sigma$. Generally speaking, the more spread-out the data are, the larger advantage our algorithm has: when $\sigma$ is very small, the performance of all three algorithms are very close to each other on both objectives. This is expected as all the requests are highly concentrated around only $Z=10$ centers, which makes good choices of routes very limited. On the other hand, when $\sigma$ is larger (which has the similar effect as larger $Z$ with fixed $\sigma$), our algorithm exhibits clear superiority, with over 50\% less total in-transit latency than \FESI or \PGP, and 30\% less total travel distance. \HRA is also much less sensitive to the change of $Z$ or $\sigma$ compared with the two baselines.

\vspace{0.5\baselineskip}
\noindent\textbf{Realworld data.}
In terms of the effect of $n,m,\lambda$, results on the two realworld datasets as well as the synthetic datasets are quite similar, so we will take one of them as example.
Figure~\ref{fig:nyc-nonEuclidean} shows the result on the NYC dataset. Both the total travel distance and in-transit latency grow with $n$, as expected, and the gap between \HRA and the baselines also grows with $n$ (the 1st column of Figure~\ref{fig:nyc-nonEuclidean}). Notice that \FESI performs better than \PGP in terms of in-transit latency, because it explicitly optimizes makespan and results in shorter per-vehicle trips.

\begin{figure}[!htbp]
\captionsetup[subfigure]{labelformat=empty}
   \centering
   \subfloat
   {
     \includegraphics[width=0.9\textwidth]{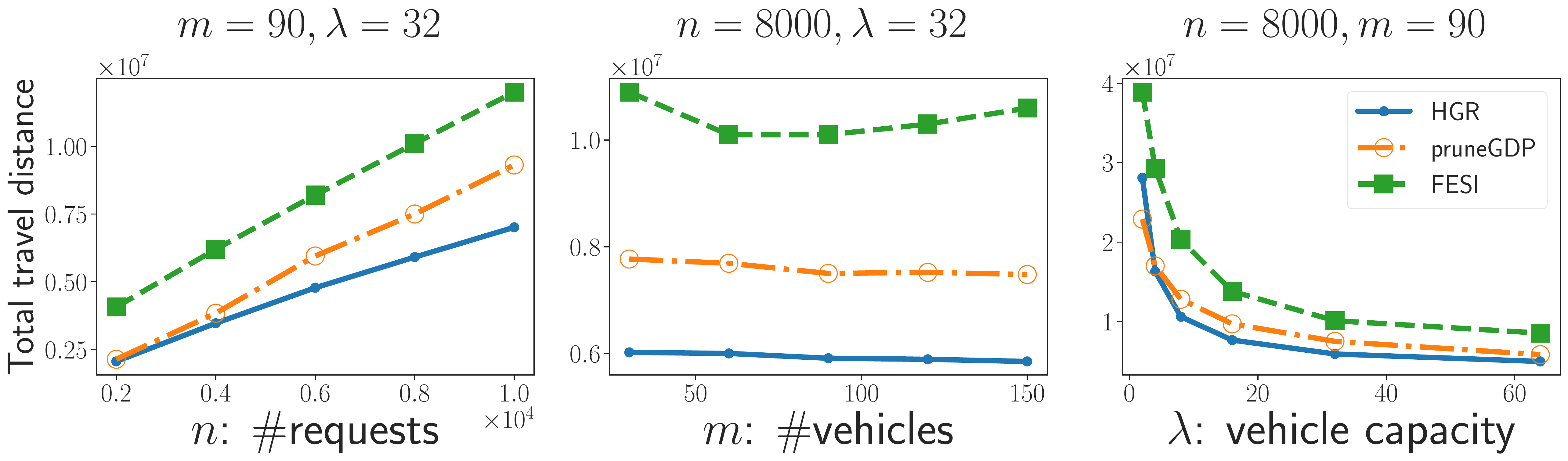}
   }
   \\
   \subfloat
   {
     \includegraphics[width=0.9\textwidth]{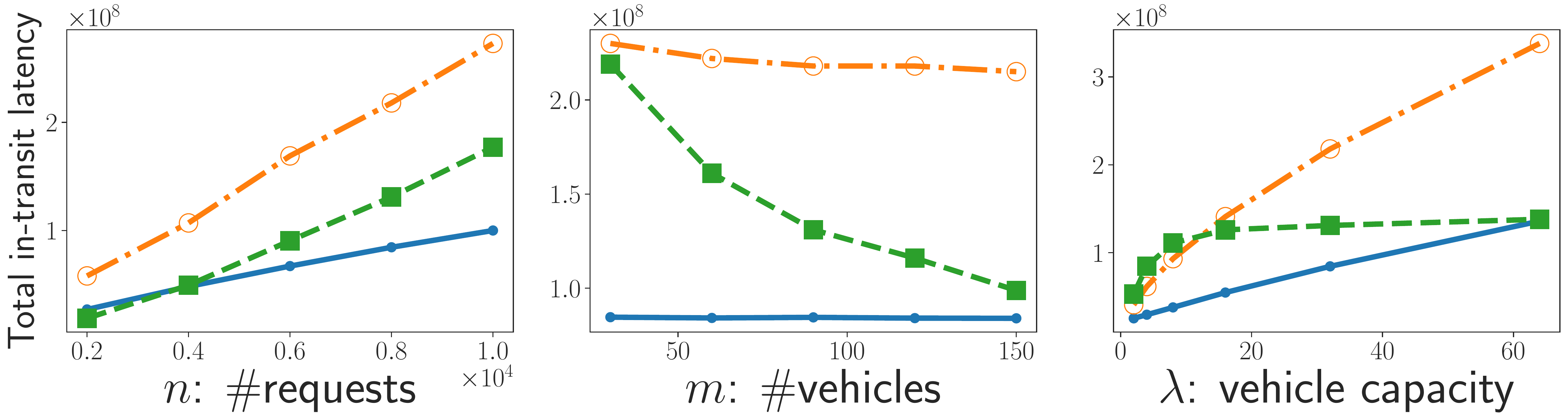}
   }
   \caption{Results on NYC dataset.}
   \label{fig:nyc-nonEuclidean}
\end{figure}

\begin{figure}[!htbp]
\captionsetup[subfigure]{labelformat=empty}
   \centering
   \subfloat
   {
     \includegraphics[width=0.9\textwidth]{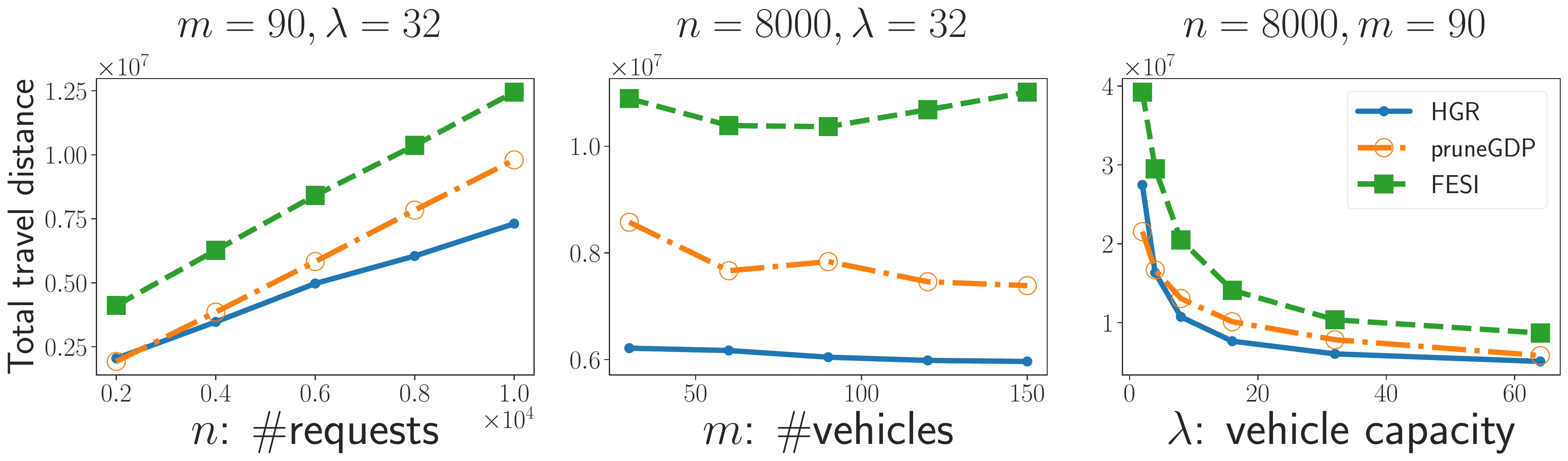}
   }
   \\
   \subfloat
   {
     \includegraphics[width=0.9\textwidth]{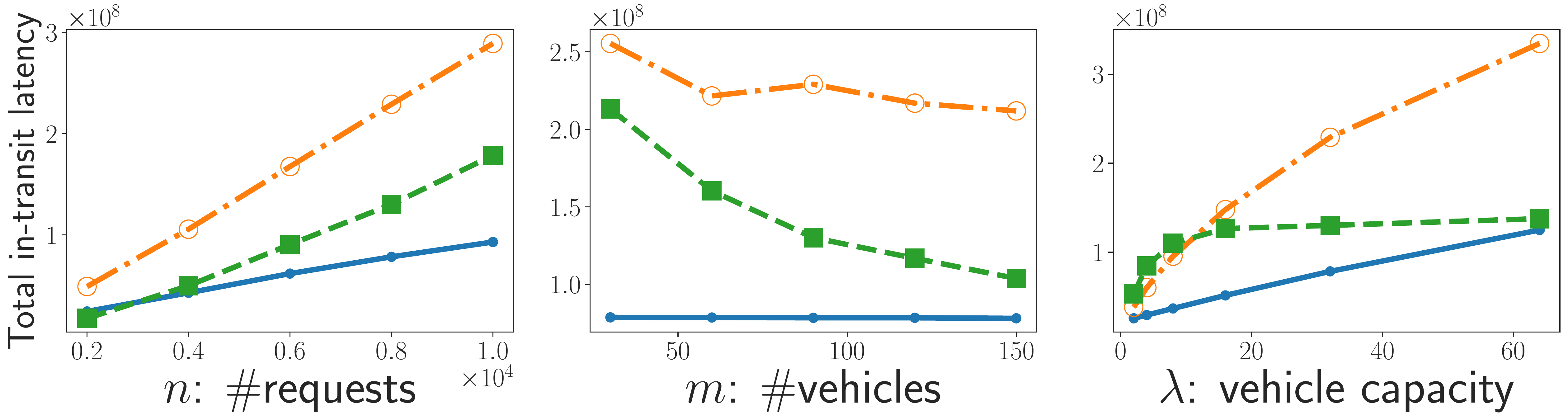}
   }
   \caption{Results on the SFO dataset.}
   \label{fig:sfo-nonEuclidean}
\end{figure}
The more interesting part is when we fix $n$ and $\lambda$, and vary the number of vehicles $m$ (the 2nd column of Figure~\ref{fig:nyc-nonEuclidean}). Notice that the latency of our algorithm (\HRA) is almost unaffected by $m$: This is because the first phase (i.e., Algorithm~\ref{alg:partition_notexact}) of \HRA is independent of vehicle locations. After the requests have been partitioned into groups, the in-transit latency is essentially determined no matter how we assign vehicles to these groups. A surprising fact is that the total travel distance of \HRA is also little affected by $m$. After inspecting the actual routes generated by \HRA, we find that although there are many vehicles available, the minimum spanning forest $\mathcal{F}$ found in Algorithm~\ref{alg:assignandroute} uses only a few vehicles.


The 3rd column of Figure~\ref{fig:nyc-nonEuclidean} shows the result where we fix $n$ and $m$ and vary $\lambda$. Larger $\lambda$ generally leads to less travel distance, since larger capacity allows more flexible choices of routes. Our algorithm still outperforms the two baselines in both objectives, but from the last plot one can expect \FESI to have better latency when $\lambda$ is larger: this is, again, because \FESI aims to optimize makespan, thus larger $\lambda$ does not necessarily lead to longer per-vehicle routes, while it is the opposite for \HRA and \PGP.

We remark that, for both the synthetic datasets and the real-world datasets, the effect of varying $n,m,\lambda$ are much the same.

\vspace{0.5\baselineskip}
\noindent{\bf Running Time Analysis.} 
In the Grouping phase, we construct at most $n^2$ minimum spanning trees in each iteration $\ell<\log \lambda$, each of them can be constructed in time $O(2^{2\ell} \log  (2^\ell))$~\cite{pettie2000optimal}; 
The most time-consuming part is step~\ref{step:cluster-matching} where we compute a minimum-cost perfect matching in each iteration, and the currently best-known algorithm takes time $O(n^3)$~\cite{gabow1990data} in a complete graph; 

Since there are $\log \lambda$ iterations, the running time of Grouping phase is $O(n^3 \log \lambda+ n^2 \lambda^2 \log \lambda)$. In the Routing phase, we find minimum spanning forest in time $O(n^2+m)$ and route the walks in time $O(mn)$. Obviously, $ m\le n$. In total, the running time is $O(n^3 \log \lambda+ n^2 \lambda^2 \log \lambda)$.

In our actual implementation of \HRA, we use The Blossom V~\cite{kolmogorov2009blossom} matching algorithm due to its widespread use in practice, in spite of having a slightly worse theoretical guarantee. 
In the experiments, for inputs consisting of 10,000 requests and 150 vehicles with capacity 64, our algorithm takes less than 800 seconds to finish. In a more practical instance with 4,000 requests and 90 vehicles with capacity 8, \HRA takes about 80 seconds. 

Although reasonably fast, our basic implementation of \HRA turns out be significantly slower than \FESI and \PGP, both of which have only an $O(n^2)$ dependence on the number of requests $n$. 
For example, on the largest input mentioned above ($n$=10,000, $m$=150, $\lambda$=64), our algorithm (\HRA) is $20$ times slower than \FESI and \PGP. In Section~\ref{app:sec-scalability}, we show how to implement a much more scalable version of \HRA by replacing several components of the vanilla algorithm with their \emph{approximate versions} while not sacrificing a lot in the quality of the solution.

%% file: scalability.tex
\section{Scalability with large instances}\label{app:sec-scalability}
\subsection{The \HRAapprox Algorithm}
As discussed at the end of Section~\ref{sec:experiments}, the dominating factor of the running time is from the matching step (Line~\ref{step:cluster-matching} of Algorithm~\ref{alg:partition_notexact}). Although in theory there exist near linear-time (i.e., $O(n^{2+\epsilon})$, because the input here can have $O(n^2)$ edges) algorithms that output a near-optimal perfect matching\cite{DuanPettie10}, which in principle can reduce our algorithm's running time to $O(n^2)$ (assuming $\lambda$ being a relatively small constant), with some negligible loss in approximation ratio. But such matching algorithms are sophisticated and not easy to implement in practice.

We therefore resort to simpler \emph{approximations}: instead of finding the min-cost perfect matching, we find a non-optimal matching using a ``bucketing'' method (details see below). Besides, the ``edge cost'' $w$ (Definition~\ref{def:costfunction}) used in our graph requires computing multiple MST over the two clusters, which is quite costly. Specifically, those MST computations come from evaluating $w_1(\cdot, \cdot)$ between groups. We therefore replace $w_1$ with a simpler cost function that approximates it. Specifically, we implement the two following approximate versions of \HRA:
\begin{itemize}
\item \HRAwone: We define a new cost $w'_1$ between any two groups $X$ and $X'$ as $$w'_1(X,X'):=\min_{r_i\in X, r_j\in X'}  d(s_i, s_j) + \min_{r_i\in X, r_j\in X'}  d(t_i, t_j).$$ The new algorithm \HRAwone still finds the minimum-cost \emph{perfect} matching (like \HRA), but uses $w'_1$ in place of $w_1$.

\item \HRAapprox: This algorithm builds upon \HRAwone. In addition to using $w'_1$ in place of $w_1$, \HRAapprox finds a perfect matching using a "bucketing" heuristic: Suppose the largest edge cost is $\Delta$. We first divide all edges into $O(\log \Delta)$ buckets, where the $i$-th bucket contains all edges with cost $[(1+\delta)^{i-1}, (1+\delta)^i)$, where $\delta$ is a small constant. Then for each bucket we compute a \emph{maximal} matching using only the edges of the bucket.

The maximal matching is computed using a simple greedy method: start with an empty matching, the algorithm greedily chooses the lowest-cost edge that is disjoint with the current matching and include to the solution.

It is straightforward to implement this heuristic in $O(n^2\log \Delta)$ time.
\end{itemize}

The \HRAwone algorithm still has an $O(n^3)$ dependence with $n$ since it needs to find a perfect matching, though it avoids the $O(n^2\lambda^2\log\lambda)$ additive factor. The \HRAapprox algorithm instead runs in $O(n^2\log \Delta)$ time.

\subsection{Experimental Evaluation}
We inherit most of the experiment settings from Section~\ref{sec:experiments}, and test \HRAapprox on much larger instances. The results are quite similar, so we only plot some representative results on the NYC dataset. We are still comparing \HRAapprox with \PGP and \FESI. To examine how the approximation affects solution quality, we also include comparison with the original \HRA in some smaller instances as in~\ref{sec:experiments}. (The original \HRA algorithm is too slow for instance as large as $n=10^5$ and does not end in reasonable time). 
\begin{center}
\begin{tabular}{cc} 
 \toprule
 $n$ (\#requests) & 2.5, 5, \textbf{10} $(\times10^4)$\\
 $m$ (\#vehicles) &  1, 2.5, 5, \textbf{10}, $(\times10^3)$ \\
 $\lambda$ (vehicle capacity) & 2, 4, 8, 16, \textbf{32}, 64 \\
 \bottomrule
\end{tabular}
\end{center}

\begin{figure}[!htbp]
\captionsetup[subfigure]{labelformat=empty}
   \centering
   \subfloat
   {
     \includegraphics[width=0.9\textwidth]{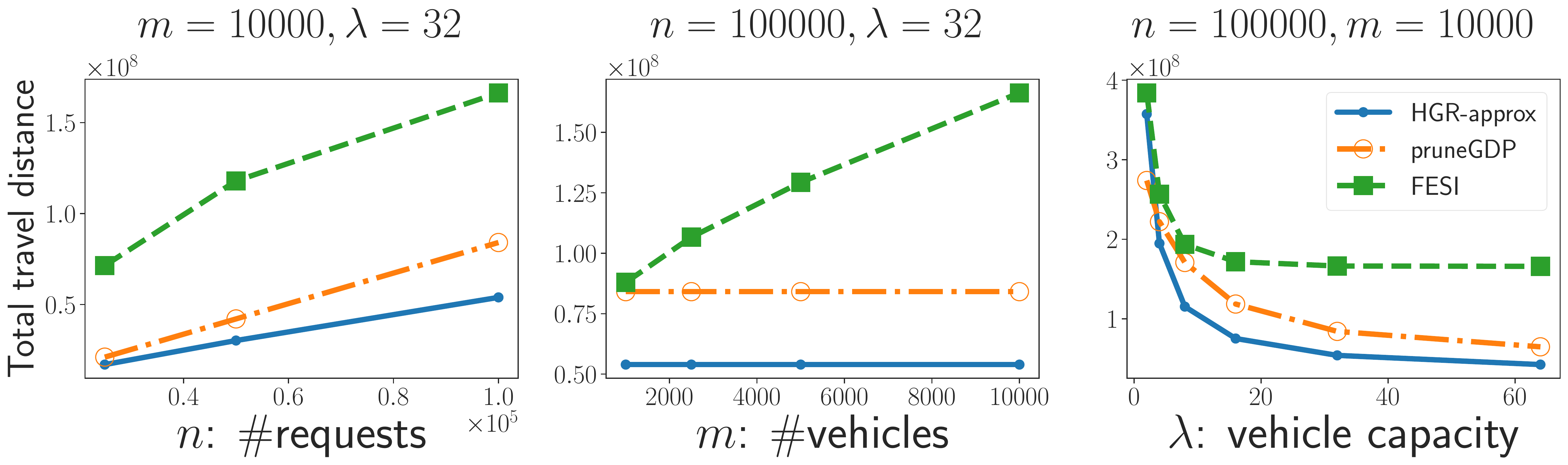}
   }
   \\
   \subfloat
   {
     \includegraphics[width=0.9\textwidth]{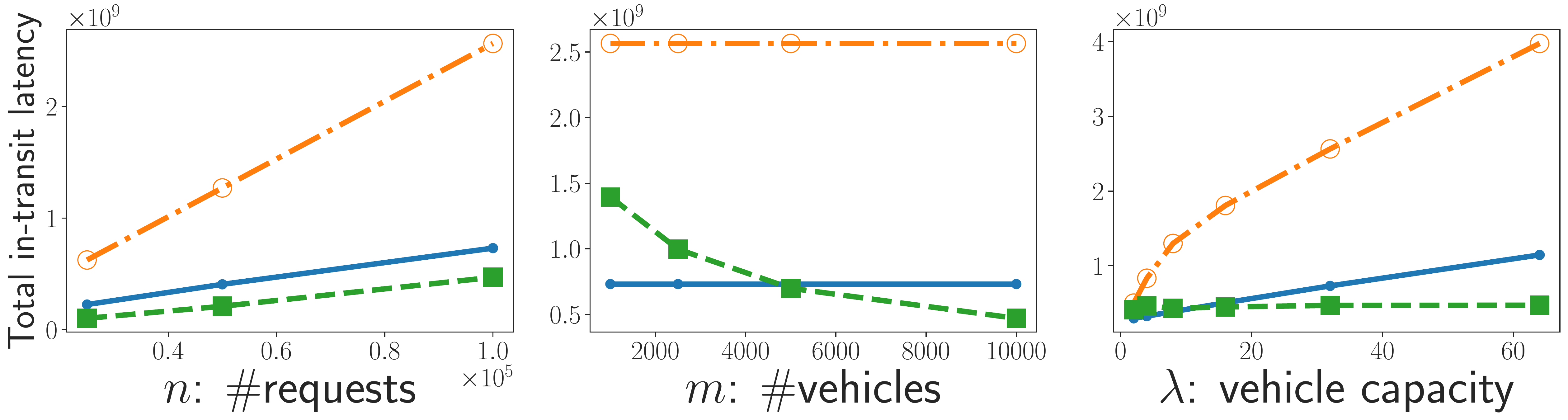}
   }
   \\
   \subfloat
   {
     \includegraphics[width=0.9\textwidth]{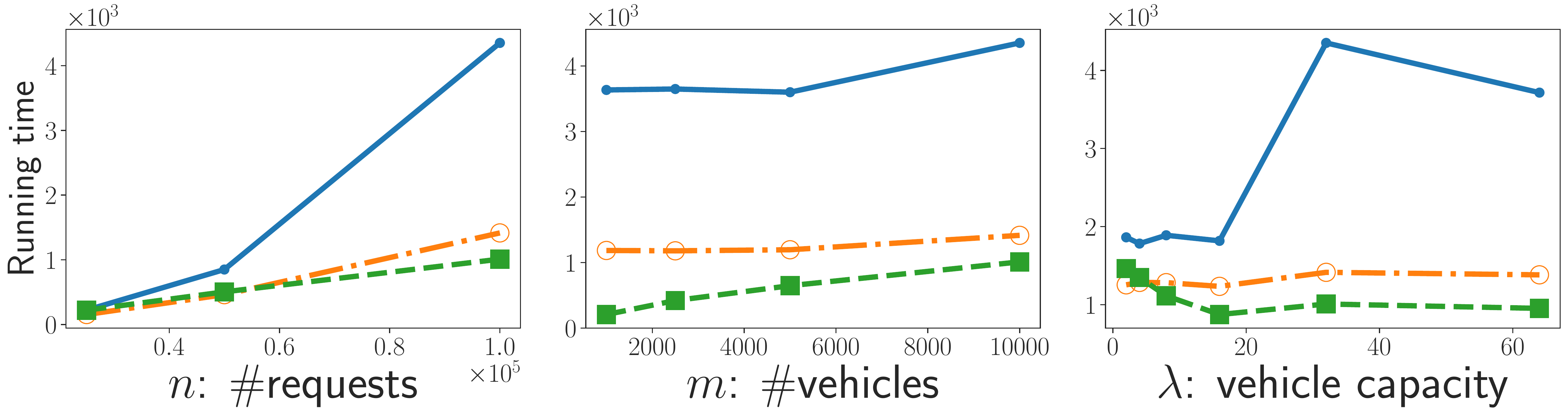}
   }
   \caption{Results on large-scale inputs from the NYC dataset.}
   \label{fig:nyc-nonEuclidean-large}
\end{figure}
Figure~\ref{fig:nyc-nonEuclidean-large} shows the results on large instances with up to 100k requests and 10k drivers. One can see that \HRAapprox still achieves best total travel distance (the first row), while having comparable in-transit latency (the second row) with the state-of-the-art benchmark (FESI). Note that \FESI explicitly optimize makespan (the largest travel distance of all vehicles), which often also leads to small in-transit distance since the solution is formed by many short trips, but the total distance can be very large. The running time (the third row) of \HRAapprox is, however, still larger than the two benchmarks, though the difference ($\sim4\times$) is much smaller that of the original \HRA, which does not even finish in a reasonable time on such-sized instances. This indicates that our algorithm has the potential to scale to large instances.

\begin{figure}[!htbp]
\captionsetup[subfigure]{labelformat=empty}
   \centering
   \subfloat
   {
     \includegraphics[width=0.9\textwidth]{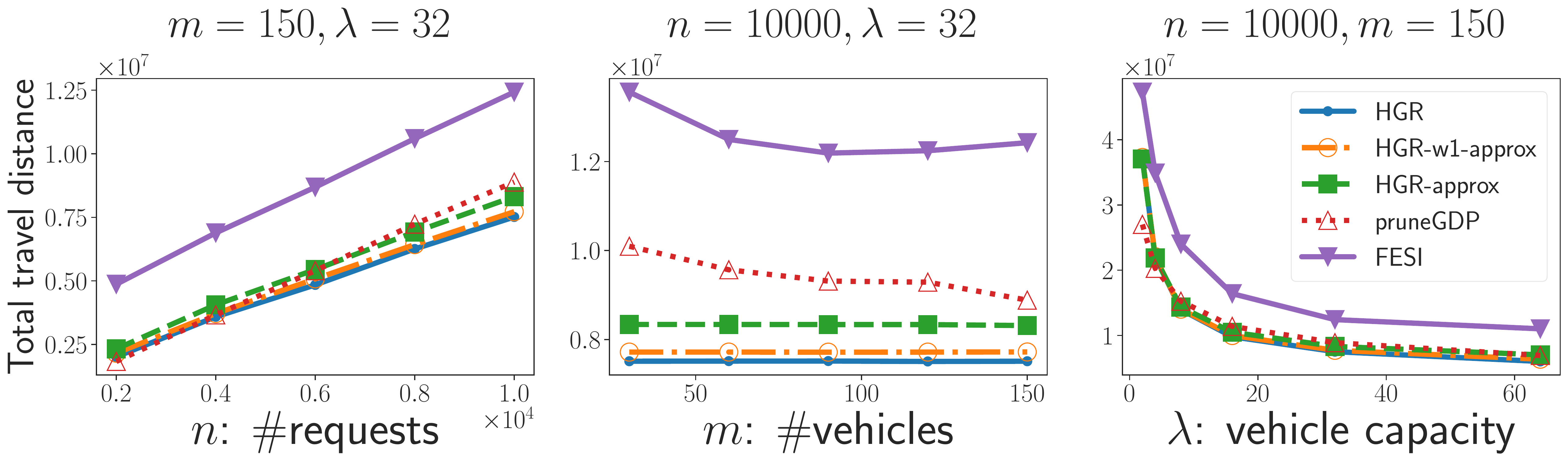}
   }
   \\
   \subfloat
   {
     \includegraphics[width=0.9\textwidth]{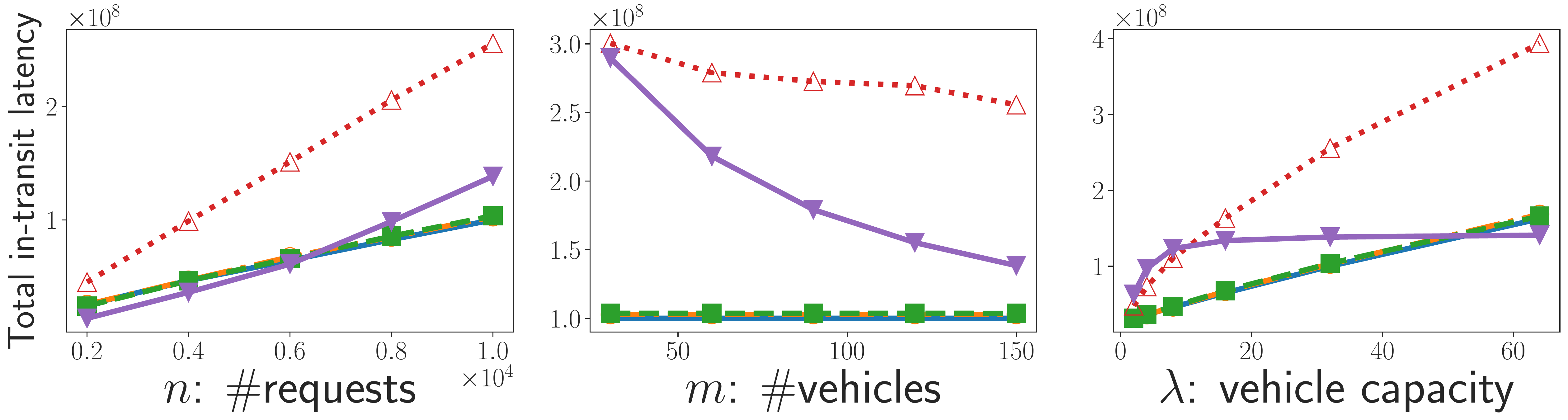}
   }
   \\
   \subfloat
   {
     \includegraphics[width=0.9\textwidth]{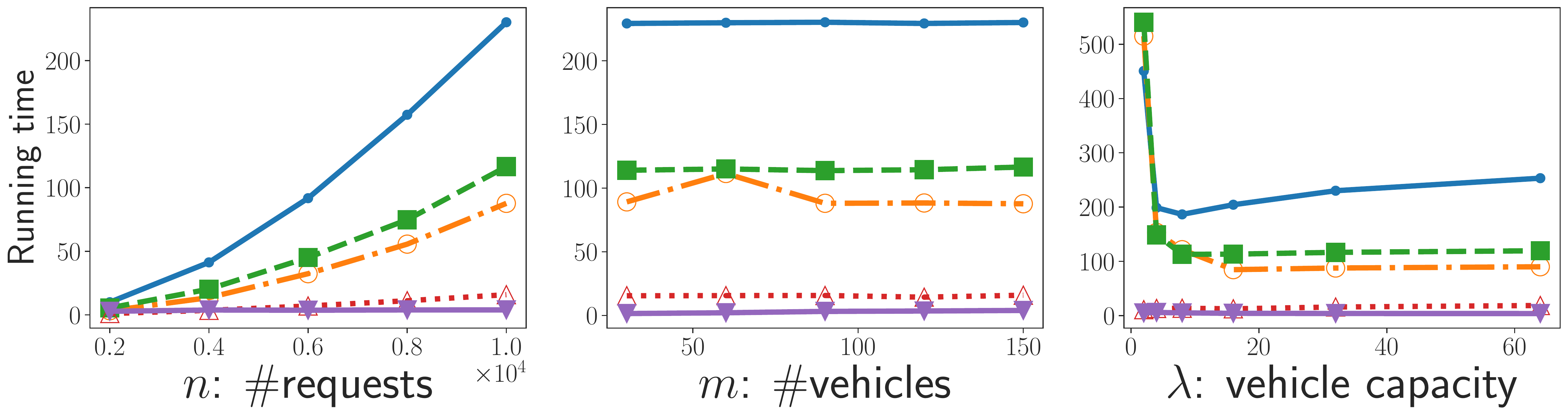}
   }
   \caption{Results on inputs from NYC dataset of the same scale as Section~\ref{sec:experiments}.}
   \label{fig:nyc-nonEuclidean-small-with-HGR-approx}
\end{figure}

Figure~\ref{fig:nyc-nonEuclidean-small-with-HGR-approx} illustrates the performance of \HRAapprox on smaller instances, with comparison to the original \HRA algorithm and \HRAwone. One can see that the approximations do lead to (slight) performance degradation: In terms of total travel distance (which is the objective our algorithm set to optimize), \HRAapprox is slightly worse than \HRAwone, which is slightly worse than \HRA. The effect on in-transit latency is even less obvious. The approximation also greatly reduces the running time of \HRA, though they are still higher than \PGP and \FESI.

In summary, we are able to accelerate \HRA significantly using some straightforward approximation or heuristics, sacrificing the solution quality only mildly. The main ingredient of algorithm --- hierarchically grouping requests --- is quite flexible and provides a good starting point to apply other routing methods. We believe our algorithm can be made even faster with cleverer optimizations, but this is beyond the scope of this paper.


%% file: conclusion.tex
\section{Discussion}

In this paper we propose an algorithm for the multi-vehicle Dial-a-Ride problem with the objective to optimize the total travel distance of all vehicles. The ${O}(\sqrt{\lambda}\log n)$ approximation ratio of our algorithm matches that of the best known algorithm for the single-vehicle case. It is still an open problem whether this ratio can be improved even in the single-vehicle case. 
We provide three different implementations of the basic algorithm with increasing runtime efficiency. We 
experimentally demonstrate that all versions of our algorithm outperforms two recent state-of-the art heuristics for this problem on both synthetic and real world datasets. Further, we showcase scalability of the most efficient implementation to datasets of size up to 100000 requests.  

We feel that our work is a significant first step towards building theoretically sounds algorithms for multi-vehicle Dial-a-Ride that are also practical. There are several intriguing open questions that out work raises. Firstly, our algorithm does not directly handle deadline constraints that are often encountered in practice.
Note that we could heuristically incorporate deadlines using the ideas given in~\cite{ZengTC19-LMD} - first solve the problem without deadlines using \HRA and then use the insertion subroutine from~\cite{TongZZCYX18-GDP} as long as no deadline constraint is violated. However, proving similar approximation guarantees as in this work becomes much more challenging in this case. We leave this as a future research direction.  

The most efficient (approximate) version of our algorithm has an asymptotic runtime complexity of $O(n^2)$ and we also demonstrate scalability of this version experimentally. However, we believe that the runtime can be further improved by exploiting the inherent parallelism of many of the steps and utilizing algorithms developed in the recently popular Map-Reduce models of computation~\cite{LattanziMSV11-MPC}. We leave this as our second open problem.